\documentclass[sigconf, authorversion=true,nonacm=true]{acmart} 

\usepackage{tikz-qtree}
\usepackage{algorithm}
\usepackage{algpseudocode}
\makeatletter
\renewcommand{\ALG@beginalgorithmic}{\footnotesize}
\makeatother
\usepackage{graphicx}
\usepackage{subcaption}
\usepackage{hyperref}
\usepackage{amsmath}
\usepackage{relsize}
\usepackage{enumitem}
\usepackage{bold-extra}
\usepackage{cleveref}
\usepackage{thmtools,thm-restate}

\usepackage{multicol}
\numberwithin{equation}{section}

\def\draft{0}  
\if\draft 1
\usepackage{showkeys}
\newcommand{\mycomment}[3]{{\color{#2}{[\bf{#1: #3}]}}}

\newcommand{\here}[1]{\bf{[[[#1]]]}}
\else
\newcommand{\mycomment}[3]{}

\newcommand{\here}[1]{}  
\fi

\newcommand{\Eric}[1]{\mycomment{Eric}{magenta}{#1}}
\newcommand{\Hossein}[1]{\mycomment{Hossein}{blue}{#1}}

\newcommand{\node}{node}
\newcommand{\nodes}{nodes}
\newcommand{\block}{block}
\newcommand{\blocks}{blocks}
\newcommand{\ordering}{ordering}

\algnewcommand\algorithmicforeach{\bf{for each}}
\algdef{S}[FOR]{ForEach}[1]{\algorithmicforeach\ #1\ \algorithmicdo}

\algdef{S}[FUNCTION]{Function}
   [3]{{\typ{{#1}}} {\op{#2}}\ifthenelse{\equal{#3}{}}{}{(\var{#3})}}
  
\algdef{E}[FUNCTION]{EndFunction}
   [1]{\algorithmicend\ \op{{#1}}}

\algrenewcommand\Call[2]{\op{#1}\ifthenelse{\equal{#2}{}}{}{(#2)}}
\usepackage{eqparbox}

\newcommand\keywordfont{\sffamily\bfseries}
\algrenewcommand\algorithmicend{{\keywordfont end}}
\algrenewcommand\algorithmicfor{{\keywordfont for}}
\algrenewcommand\algorithmicforeach{{\keywordfont for each}}
\algrenewcommand\algorithmicdo{{\keywordfont do}}
\algrenewcommand\algorithmicuntil{{\keywordfont until}}
\algrenewcommand\algorithmicfunction{{\keywordfont function}}
\algrenewcommand\algorithmicif{{\keywordfont if}}
\algrenewcommand\algorithmicthen{{\keywordfont then}}
\algrenewcommand\algorithmicelse{{\keywordfont else}}
\algrenewcommand\algorithmicreturn{{\keywordfont return}}

\newcommand{\sub}[1]{\textsubscript{#1}}
\renewcommand{\tt}[1]{\texttt{#1}}

\renewcommand{\it}[1]{\textit{#1}}

\renewcommand{\bf}[1]{\textbf{#1}}
\newcommand{\nf}[1]{{\normalfont{\texttt{#1}}}}
\newcommand{\cmt}[1]{\Comment{#1}}
\newcommand{\head}{\fld{head}}
\newcommand{\size}{\fld{size}}
\newcommand{\eleft}{\fld{end\sub{left}}}
\newcommand{\eright}{\fld{end\sub{right}}}
\newcommand{\edir}{\fld{end\sub{dir}}}

\newcommand{\hangbox}[1]{\parbox[t]{\dimexpr\linewidth-\algorithmicindent\relax}{%
    \setlength{\hangindent}{\algorithmicindent}%
    #1}}
\newcommand{\spac}{\medskip}
\newcommand{\assign}{\mbox{:=}}

\newcommand{\var}[1]{\mbox{\textnormal{\it{#1}}}}
\newcommand{\fld}[1]{\mbox{\it{#1}}}
\newcommand{\typ}[1]{\mbox{\sf #1}}
\newcommand{\ceil}[1]{\lceil #1 \rceil}
\newcommand{\nl}{\mbox{\sf null}}
\newcommand{\tr}{\mbox{\sf true}}
\newcommand{\fa}{\mbox{\sf false}}
\newcommand{\opa}[2]{\op{#1(\var{#2})}}
\newcommand{\op}[1]{{\sf #1}}
\newcommand{\opemph}[1]{{\sf \bf{#1}}}
\newcommand{\com}{$\triangleright$}
\newcommand{\linecomment}[1]{\com\ #1}
\newcommand{\gc}[1]{\color{blue}#1\color{black}}

\usepackage{amsmath,amsthm}
\newtheorem{theorem}{Theorem}
\newtheorem{lemma}[theorem]{Lemma}
\newtheorem{mytheorem}[theorem]{Theorem}
\newtheorem{corollary}[theorem]{Corollary}
\newtheorem{proposition}[theorem]{Proposition}

\newtheorem{observation}[theorem]{Observation}
\newtheorem{invariant}[theorem]{Invariant}
\newtheorem{claim}{Claim}[theorem]

\newtheorem{innercustomlemma}{Lemma}
\newenvironment{customlemma}[1]
  {\renewcommand\theinnercustomlemma{#1}\innercustomlemma}
  {\endinnercustomlemma}

\newtheorem{innercustomprop}{Proposition}
\newenvironment{customprop}[1]
  {\renewcommand\theinnercustomprop{#1}\innercustomprop}
  {\endinnercustomprop}

\AtBeginDocument{%
  }

\setcopyright{none}  

\acmConference[PODC 2023]{ACM Symposium on Principles of Distributed Computing}{June 19--23,
  2023}{Orlando, FL}

\title{A Wait-free Queue with Polylogarithmic Step Complexity}

\author{Hossein Naderibeni}
\affiliation{%
  \institution{York University}
  \streetaddress{P.O. Box 1212}
  \city{Toronto}
  \state{Ontario}
  \country{Canada}
  \postcode{43017-6221}
}
\author{Eric Ruppert}
\affiliation{%
  \institution{York University}
  \streetaddress{P.O. Box 1212}
  \city{Toronto}
  \state{Ontario}
  \country{Canada}
  \postcode{43017-6221}
}
\orcid{0000-0001-5613-8701}

\copyrightyear{2023}
\acmYear{2023}
\setcopyright{acmlicensed}\acmConference[PODC '23]{ACM Symposium on Principles of Distributed Computing}{June 19--23, 2023}{Orlando, FL, USA}
\acmBooktitle{ACM Symposium on Principles of Distributed Computing (PODC '23), June 19--23, 2023, Orlando, FL, USA}
\acmPrice{15.00}
\acmDOI{10.1145/3583668.3594565}
\acmISBN{979-8-4007-0121-4/23/06}

\ccsdesc[500]{Theory of computation~Data structures design and analysis}
\ccsdesc[500]{Theory of computation~Concurrent algorithms}

\keywords{concurrent data structures, wait-free queues}

\begin{document}

\begin{abstract}
We present a novel linearizable wait-free queue implementation using single-word CAS instructions.  
Previous lock-free queue implementations from \op{CAS} all have amortized step complexity of 
$\Omega(p)$ per operation in worst-case executions, where $p$ is the number of processes that 
access the queue.  Our new wait-free queue  takes $O(\log p)$ steps per enqueue and 
$O(\log^2 p +\log q)$ steps per dequeue, where $q$ is the size of the queue.  
A bounded-space version of the implementation has $O(\log p \log(p+q))$ amortized step complexity  per operation.
\end{abstract}

\maketitle

\Eric{Remove references to appendix and replace with tech report}


\section{Introduction}

There has been a great deal of research in the past several decades on the design of shared queues.
Besides being a fundamental data structure, queues are used in
significant concurrent applications, including OS kernels \cite{MP91}, memory management \cite{BBFRSW21}\here{try to find an older, more canonical reference},
synchronization \cite{KAE23},\here{is this a good citation for this?}
and sharing resources or tasks.
We focus on shared queues that are \emph{linearizable} \cite{HW90}, meaning that operations
appear to take place atomically, and \emph{lock-free}, meaning that some operation on the queue
is guaranteed to complete regardless of how asynchronous processes are scheduled to take~steps.

The lock-free MS-queue of Michael and Scott \cite{MS98} is a classic shared queue implementation.
It uses a singly-linked list with pointers to the front and back nodes.
To dequeue or enqueue an element, the front or back pointer is updated by a 
compare-and-swap (CAS) instruction.
If this CAS fails, the operation must retry.
In the worst case, this means that each successful CAS may cause all other processes to
fail and retry, leading to an amortized step complexity of $\Omega(p)$ per operation in a system of $p$ processes.
(To measure amortized step complexity of a lock-free implementation, we consider all possible finite executions
and divide the number of steps in the execution by the number of operations  in the execution.)
Numerous papers have suggested modifications to the MS-queue \cite{DBLP:conf/opodis/HoffmanSS07,DBLP:conf/podc/KoganH14,DBLP:conf/ppopp/KoganP11,DBLP:journals/dc/Ladan-MozesS08,MKLLP22,DBLP:conf/spaa/MoirNSS05,RC17}, but 
all still have $\Omega(p)$ amortized step complexity as a result of
contention on the front and back of the queue.
Morrison and Afek \cite{DBLP:conf/ppopp/MorrisonA13} called this the \emph{CAS retry problem}.
The same problem occurs in array-based implementations of queues \cite{DBLP:conf/iceccs/ColvinG05,DBLP:conf/icdcn/Shafiei09,DBLP:conf/spaa/TsigasZ01,DBLP:conf/opodis/GidenstamST10}.
Solutions that tried to sidestep this problem using fetch\&increment \cite{DBLP:conf/ppopp/MorrisonA13,DBLP:conf/ppopp/YangM16,Nik19,10.1145/3490148.3538572}
rely on slower mechanisms to handle worst-case executions and still have $\Omega(p)$ step complexity.

Many concurrent data structures that keep track of a set of elements also have an $\Omega(p)$ term in their step complexity, as observed by Ruppert \cite{Rup16}.
For example, lock-free lists \cite{FR04,Sha15}, stacks \cite{Tre86} and search trees \cite{EFHR14} 
have an $\Omega(c)$ term in their step complexity, where $c$ represents contention,
the number of processes that access the data structure concurrently, which can be $p$ in the worst case.
Attiya and Fouren \cite{DBLP:conf/opodis/AttiyaF17} proved 
that amortized $\Omega(c)$ steps per operation are indeed necessary
for any CAS-based implementation of a lock-free bag data structure, which provides operations
to insert an element or remove an arbitrary element (chosen non-deterministically).
Since a queue trivially implements a bag, this lower bound also applies to queues.
Although this might seem to settle the step complexity of lock-free queues, the lower bound
holds only if $c$ is $O(\log\log p)$ so it should be stated more precisely as
an amortized bound of $\Omega(\min(c,\log\log p))$ steps per operation.

We exploit this loophole.  We show  it is, in fact, possible for a linearizable queue
to have step complexity sublinear in $p$.
Our queue is the first whose step complexity  is polylogarithmic in $p$ and in $q$, the number of elements in the queue.
It is \emph{wait-free}, meaning that every operation is guaranteed to complete within a finite number of its own steps.
For ease of presentation, we first give an unbounded-space construction where enqueues take $O(\log p)$ steps and
dequeues take $O(\log^2 p + \log q)$ steps,
and then modify it to bound the space
while  having $O(\log p\log( p+ q))$ amortized step complexity  per operation.
Moreover, each operation does $O(\log p)$ CAS instructions in the worst case, whereas previous
lock-free queues use 
$\Omega(p)$ CAS instructions, even in an amortized sense.
Both versions of our queue use single-word CAS on reasonably-sized 
words.
We assume that a word is large enough to store an item to be enqueued (or at least a pointer to it).  We also assume that the number of operations performed on the queue can be stored (in binary) in $O(1)$ words.
This is analogous to the assumption for the classical RAM model that the number of bits per word is logarithmic in the problem size.
For the space-bounded version, we unlink unneeded objects from our data structure.
We do not address the orthogonal problem of reclaiming memory; we assume a safe
garbage collector, such as the highly optimized one that Java provides.

Our queue uses a binary tree, called the \emph{\ordering\ tree}, where each process has its own leaf.
A process adds its operations to its leaf.
As in previous work (e.g., \cite{DBLP:conf/stoc/AfekDT95,DBLP:conf/fsttcs/JayantiP05}), operations are propagated from the leaves up to the root in a cooperative way that ensures wait-freedom
and avoids the CAS retry problem.
Operations in the root are ordered, 
and this order is used to linearize the operations and compute their responses.
\here{Either here or in related work section, talk about previous usage of ordering tree and how ours differs from it}
Explicitly storing  operations in the tree nodes would be too costly.
Instead, we use a novel implicit representation of sets
of operations that allows us to quickly merge two sets from the children of a node,
and quickly access any  operation in a~set.
A preliminary version of this work appeared in~\cite{Nad22}.
\here{Maybe say a little more about techniques used in the implementation, if space permits}


\section{Related Work}

\paragraph{List-based Queues.}
The MS-queue \cite{MS98} is a lock-free queue that has stood the test of time.
The standard Java Concurrency Package includes a version of it.  
See \cite{MS98} for a survey of the early history of concurrent queues.
As mentioned above, the MS-queue suffers from the CAS retry problem because of contention at the front and back of the queue.
Thus, it is lock-free but not wait-free and has an amortized step complexity of $\Theta(p)$ per operation.


Many papers have described ways to reduce contention in the MS-queue.
Moir et al.~\cite{DBLP:conf/spaa/MoirNSS05} 
added an elimination array that allows an enqueue to pass its enqueued value directly
to a concurrent dequeue when the queue is empty.
However, when there are $p$ concurrent enqueues (and no dequeues), the CAS retry problem
is still present.
The baskets queue of
Hoffman, Shalev, and Shavit~\cite{DBLP:conf/opodis/HoffmanSS07} 
attempts to reduce contention by grouping concurrent enqueues into baskets.
An enqueue that fails its CAS is put in the basket with the enqueue that succeeded.
Enqueues within a basket order  themselves without having to access the back of the queue.
However, if $p$ concurrent enqueues are in the same basket
the CAS retry problem occurs when they order themselves using CAS instructions.
Both modifications still have $\Omega(p)$ amortized step complexity.

Kogan and Herlihy \cite{DBLP:conf/podc/KoganH14} improved the MS-queue's performance using \emph{futures}.
Operations return future
objects instead of responses. Later, when an operation's response is needed, it
is evaluated using the future object.
This allows batches of enqueues or dequeues to be done at once on an MS-queue.
However, the implementation satisfies a weaker correctness condition than linearizability.
Milman-Sela et al.~\cite{MKLLP22} extended this approach to allow batches
to mix enqueues and dequeues.
\here{Is this worth saying:
They use some properties of the queue size before and after a batch, similar to a part of our work.}
In the worst case, where operations require their response right away,
batches have size 1, and both of these implementations behave like a standard MS-queue.

In the MS-queue, an enqueue requires two CAS steps.
Ladan-Mozes and Shavit~\cite{DBLP:journals/dc/Ladan-MozesS08}
presented an optimistic  queue implementation
that uses a doubly-linked list to reduce the number of
\texttt{CAS} instructions to one in the best case. 
Pointers in the doubly-linked list can be inconsistent, but are fixed when necessary by traversing the list.
This fixing is rare in practice, but it yields an amortized complexity of $\Omega(qp)$ 
steps per operation in the worst case.

Kogan and Petrank~\cite{DBLP:conf/ppopp/KoganP11} 
used Herlihy's helping
technique~\cite{10.1145/114005.102808} to make the MS-queue
wait-free.
Then, they introduced the 
fast-path slow-path methodology \cite{10.1145/2370036.2145835} for making data structures wait-free:
the fast path has good performance and the slow path guarantees termination.
They applied their methodology to combine the MS-queue (as the fast path)
with their wait-free queue (as the slow path).
Ramalhete and Correia \cite{RC17} added a different helping mechanism to the MS-queue.
Although these approaches can perform well in practice,
the amortized step complexity remains~$\Omega(p)$. 
\Hossein{I get what this sentence means, but from a high overview i think it is inaccurate. Since the step complexity of MS-queue is infinite, so nothing can be added to it and a wait-free queue is always better than a lock-free queue in the worst-case.}
\Eric{I modified the wording.  I also added the material about fast-path slow-path here (since I had forgotten that paper also applies it to a queue), and fixed the paragraph about FPSP method below when discussing array-based queues.}

\paragraph{Array-Based Queues.}
Arrays can be used to implement queues with bounded capacity \cite{DBLP:conf/iceccs/ColvinG05,DBLP:conf/icdcn/Shafiei09,DBLP:conf/spaa/TsigasZ01}.  
Dequeues and enqueues update
indices of the front and back elements using CAS instructions.
Gidenstam, Sundell, and Tsigas~\cite{DBLP:conf/opodis/GidenstamST10} avoid
the capacity constraint by using a linked list of arrays.
These solutions also use $\Omega(p)$ steps per operation due to the CAS retry problem.

Morrison and Afek \cite{DBLP:conf/ppopp/MorrisonA13} also used a linked list of (circular) arrays.
To avoid the CAS retry problem, concurrent operations try to claim spots in an array using fetch\&increment instructions.
(It was shown recently that this implementation can be modified to use single-word CAS instructions rather than double-width CAS \cite{RK23}.)
If livelock between enqueues and a racing dequeue prevent enqueues from claiming a spot,
the enqueues fall back on using a CAS to add a new array to the linked list, 
and the CAS retry problem reappears.
This approach is similar to the fast-path slow-path methodology \cite{10.1145/2370036.2145835}.
Other array-based queues \cite{Nik19,10.1145/3490148.3538572,DBLP:conf/ppopp/YangM16} also used this methodology.
In worst-case executions that use the slow path,
they also take $\Omega(p)$ steps per operation, 
due either to the CAS retry problem or helping mechanisms.



\paragraph{Universal Constructions.}
One can also build a queue using a universal construction \cite{10.1145/114005.102808}.
Jayanti \cite{DBLP:conf/podc/Jayanti98a} observed that
the universal construction of Afek, Dauber, and
Touitou~\cite{DBLP:conf/stoc/AfekDT95} can be modified to use $O(\log p)$ steps per operation, 
assuming that words can store $\Omega(p \log p)$ bits. 
(Thus, in terms of operations on reasonably-sized $O(\log p)$-bit words, their construction would take $\Omega(p\log p)$ steps per operation.)
Fatourou and Kallimanis \cite{FK14} used their own universal construction based on fetch\&add and LL/SC instructions
to implement a queue, but its step complexity is also $\Omega(p)$.
\here{Double check this.}

\paragraph{Restricted Queues.}
David gave the first sublinear-time queue
\cite{DBLP:conf/wdag/David04}, but it works only for a single enqueuer.
It uses fetch\&increment and swap  instructions and takes $O(1)$ steps per operation, but
uses unbounded memory.  Bounding the space increases the steps per operation to $\Omega(p)$.
Jayanti and Petrovic gave a wait-free polylogarithmic
queue~\cite{DBLP:conf/fsttcs/JayantiP05}, but only for a single dequeuer. 
Our \ordering\ tree is similar to the tree structure they use to agree on a linearization ordering.
Concurrently with our work, which first appeared in \cite{Nad22}, Johnen, Khattabi and Milani \cite{JKM23} built on \cite{DBLP:conf/fsttcs/JayantiP05} to give a wait-free queue  
that achieves $O(\log p)$ steps for enqueue operations but fails to achieve polylogarithmic step complexity for dequeues: their dequeue operations take $O(k \log p)$ steps if there are $k$ dequeuers.

\here{Should we mention https://arxiv.org/abs/2103.11926 , which seems to have a queue for 2 enqueuers and 2 dequeuers ?  May not be worth mentioning because then $p$ is a constant, so complexity in terms of $p$ doesn't make sense.}

\paragraph{Other Primitives.}
Khanchandani and Wattenhofer \cite{KW18} gave a wait-free queue 
with $O(\sqrt{p})$ step complexity using non-standard synchronization primitives
called half-increment and half-max, which can be viewed as 
double-word read-modify-write operations.
They use this as evidence that their primitives can be more efficient than CAS
since previous CAS-based queues all required $\Omega(p)$ step complexity.
Our new implementation counters this argument.

\paragraph{Fetch\&Increment Objects.}
Ellen, Ramachandran and Woelfel \cite{ERW12} gave an implementation of 
fetch\&increment objects that uses a polylogarithmic number of steps per operation.
Like our queue, they also use a tree structure similar to the universal construction of 
\cite{DBLP:conf/stoc/AfekDT95} to keep track of the operations
that have been performed.
However, our construction requires more intricate
data structures to represent sets of operations, since a queue's state cannot be represented as succinctly
as the single-word state of a fetch\&increment object.
Ellen and Woelfel \cite{10.1007/978-3-642-41527-2_20} gave an improved implementation of fetch\&increment with better step complexity.


\section{Queue Implementation} \label{DescriptQ}

\subsection{Overview}
Our \emph{\ordering\ tree} data structure is used to agree on a total ordering of the operations performed on the queue.
It is a static binary tree of height $\ceil{\log_2 p}$ with one leaf 
for each process. 
Each tree node  stores an array of \emph{blocks}, where each block represents a 
sequence of enqueues and a sequence of dequeues.
See Figure \ref{orderingtree} for an example.
In this section, we use an infinite array of blocks in each node.
Section \ref{reducing} describes how to replace the infinite array by a representation that uses bounded space.

To perform an operation on the queue, a process $P$ appends a new block containing that  
operation to the \fld{blocks} array in $P$'s leaf.
Then, $P$ attempts to propagate the operation to each node along the path from that leaf to the root of the tree.
We shall define a total order on all operations that have been propagated to the root, which 
will serve as the linearization ordering of the operations.

\Eric{Reworded following parag after Hossein's comments of Jan 4}
To propagate operations from a node $\var{v}$'s children to $\var{v}$, $P$ first observes
the blocks in both of $\var{v}$'s children that are not already in $\var{v}$,
creates a new block by combining information from those blocks, and attempts to append this 
new block to $\var{v}$'s \fld{blocks} array using a \op{CAS}.
Following \cite{DBLP:conf/fsttcs/JayantiP05}, we call this a 3-step sequence a
\op{Refresh} on $\var{v}$. 
A \op{Refresh}'s \op{CAS} may fail if there is another concurrent \op{Refresh} on~$\var{v}$.
However, since a successful \op{Refresh} propagates multiple pending operations 
from $\var{v}$'s children to $\var{v}$,
we can prove that if two \op{Refresh}es by $P$ on $\var{v}$ fail,
then $P$'s operation has been propagated to $\var{v}$ by some other process, so $P$ can continue 
onwards towards the~root.

Now suppose $P$'s operation has been propagated all the way to the root.
If $P$'s operation is an enqueue, it has obtained a place in the linearization ordering and can terminate.
If $P$'s operation is a dequeue, $P$ must use information in the tree to compute the value that the
dequeue must return.  To do this, $P$ first determines which block in the root contains its dequeue
(since the dequeue may have been propagated to the root by some other process).
$P$ does this by finding the dequeue's location in each node along the path from the leaf to the root.
Then, $P$ determines whether the queue is empty when its dequeue is linearized. 
If so, it returns \nl\ and we call it a \emph{null dequeue}.
If not, $P$ computes the rank $r$ of its dequeue among all non-null dequeues in the linearization ordering.  (We say that the $r$th element in a sequence has \emph{rank} $r$ within that sequence.)
$P$ then returns the value of the $r$th enqueue in the linearization.

We must choose what to store in each block so that the following tasks can be done efficiently.
\begin{enumerate}[label={(T\arabic*)}]
\item
\label{construct}
Construct a block for node $\var{v}$ that represents the operations in consecutive blocks in $\var{v}$'s children, as required for a \op{Refresh}.
\item
\label{findinroot}
Given a dequeue in a leaf that has been propagated to the root, find that operation's position in the root's \fld{blocks} array.
\item
\label{findrank}
Given a dequeue's position in the root, decide if it is a null dequeue (i.e., if the queue is empty when it is linearized)
or determine the rank $r$ of the enqueue whose value it returns.
\item
\label{findenqueue}
Find the $r$th enqueue in the linearization ordering.
\end{enumerate}
Since these tasks depend on the linearization ordering, we describe that ordering next.

\begin{figure*}[t]
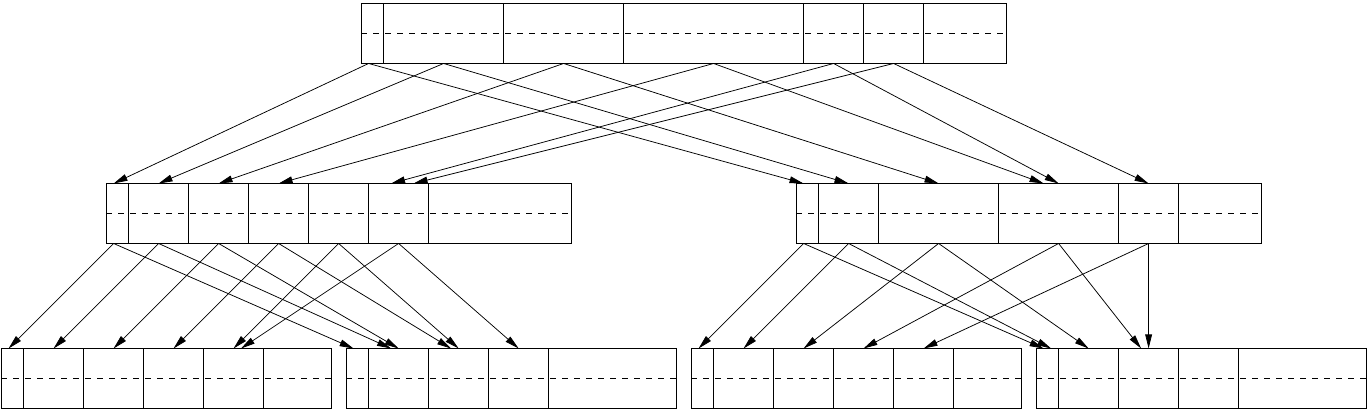
\caption{An example \ordering\ tree with four processes. 
We show explicitly the enqueue sequence and dequeue sequence represented by each block in the \fld{blocks} arrays of the seven nodes.  The leftmost element of each \fld{blocks} array is a dummy block.
Arrows represent the indices stored in \eleft\ and \eright\ fields of blocks (as described in Section \ref{sec:fields}).
The fourth process's Deq\sub{6} is still propagating towards the root.
The linearization order for this tree is
Enq(a) Enq(e) Deq\sub{2} $\mid$ Enq(b) Deq\sub{4} Deq\sub{5} $\mid$ Enq(d) Enq(f) Enq(h) Deq\sub{1} $\mid$ Enq(c) Deq\sub{3} $\mid$ Enq(g), where vertical bars indicate boundaries of blocks in the root.\label{orderingtree}}
\medskip
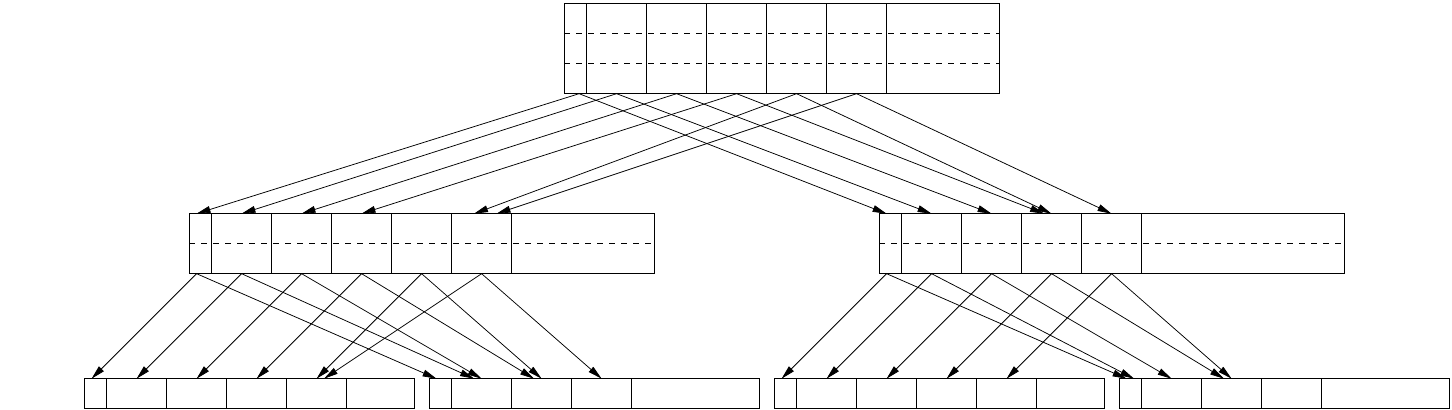
\caption{\label{implicit}The actual, implicit representation of the tree shown in Figure \ref{orderingtree}.
The leaf blocks simply show the \fld{element} field.
Internal blocks show the \fld{sum\sub{enq}} and \fld{sum\sub{deq}} fields,
and \eleft\ and \eright\ fields are shown using arrows as in Figure \ref{orderingtree}.
Root blocks also have the additional \fld{size} field.
The \fld{super} field is not shown.}
\end{figure*}

\subsection{Linearization Ordering}

Performing a double \op{Refresh} at each node along the path from the leaf to the root ensures 
a block containing the operation is appended to the root before the operation completes.
So, if an operation $op_1$ terminates before another operation $op_2$ begins, 
$op_1$ will be in an earlier block than $op_2$ in the root's blocks array.
Thus, we linearize operations according to the block they belong to in the root's array.
We can choose how to order operations in the same block, since they must be concurrent.

Each block in a leaf represents one operation.
Each block $B$ in an internal node $\var{v}$ results from merging
several consecutive blocks from each of $\var{v}$'s children.
The merged blocks in $\var{v}$'s children are called the \emph{direct subblocks} of $B$.
A block $B'$ is a \emph{subblock} of $B$ if it is a direct subblock of $B$
or a subblock of a direct subblock of $B$.
A block $B$ represents the set of operations in all of $B$'s subblocks in leaves of the tree.
The operations propagated by a \op{Refresh} are all pending when the \op{Refresh} occurs,
so there is at most one operation per process.
Hence, a block represents at most $p$ operations in total.  
Moreover, we never append empty blocks, so 
each block represents at least one operation and it follows that a block can have at most $p$ direct subblocks.

As mentioned above, we are free to order operations within a block however we like.
We order the enqueues and dequeues separately, and put the 
operations propagated from the left child before the operations from the right child.
More formally, we inductively define sequences $E(B)$ and $D(B)$ of the enqueues and dequeues
represented by a block $B$.
If $B$ is a block in a leaf representing an enqueue operation, its enqueue sequence $E(B)$ is that operation
and its dequeue sequence $D(B)$ is empty.  If $B$ is a block in a leaf representing a dequeue, $D(B)$ is that single operation and $E(B)$ is empty.
If $B$ is a block in an internal node $\var{v}$ with direct subblocks $B^L_1, \ldots, B^L_\ell$ from 
$\var{v}$'s left child
and $B^R_1,\ldots,B^R_r$ from $\var{v}$'s right child, then $B$'s operation sequences are defined by the concatenations 
\begin{eqnarray}
E(B) &=& E(B^L_1)\cdots E(B^L_\ell)\cdot E(B^R_1) \cdots E(B^R_r) \mbox{ and }\nonumber\\
D(B) &=& D(B^L_1)\cdots D(B^L_\ell)\cdot D(B^R_1) \cdots D(B^R_r)\label{defSeqs}
\end{eqnarray}
We say the block $B$ \emph{contains} the operations in $E(B)$ and $D(B)$.

When linearizing the operations propagated to the root, we must
choose how to order operations within a block.  
We choose to put
each block's enqueues before its dequeues.
Thus, if the root's blocks array contains blocks $B_1, \ldots, B_k$, the 
linearization ordering~is 
\begin{equation}
L=E(B_1)\cdot D(B_1) \cdot E(B_2) \cdot D(B_2) \cdots E(B_k) \cdot D(B_k).
\label{linearization}
\end{equation}

\subsection{Designing a Block Representation to Solve Tasks \ref{construct} to \ref{findenqueue}}
\label{sec:fields}

\renewcommand{\algorithmiccomment}[1]{\hfill\eqparbox{COMMENTSINGLE}{\com\ #1}}
\begin{figure}
\begin{algorithmic}
\setcounter{ALG@line}{1}
\Statex \linecomment{Shared variable}
\begin{itemize}
\item \typ{Node} \var{root} \Comment{root of binary tree of \tt{Node}s with one leaf per process}
\end{itemize}

\Statex \linecomment{Thread-local variable}
\begin{itemize}
\item \typ{Node} \var{leaf} \Comment{process's leaf in the tree}
\end{itemize}

\Statex $\blacktriangleright$ \typ{Node}
\begin{itemize}
\item \typ{Node} \fld{left}, \fld{right}, \fld{parent} \Comment{tree pointers initialized  when creating the tree}
\item \typ{Block}[0..$\infty$] \fld{blocks} \Comment{blocks that have been propagated to this node;}\linebreak
	\mbox{ }\Comment{\var{blocks}[0] is empty block whose integer fields are 0}
\item \typ{int} \head \Comment{position to append next \block\ to \fld{blocks}, initially 1}
\end{itemize}

\Statex $\blacktriangleright$ \typ{Block} 

\begin{itemize}
  	\item \typ{int} \fld{sum\sub{enq}, sum\sub{deq}}
  		\Comment{number of enqueues, dequeues in \fld{blocks} array}\linebreak
		\mbox{ }\Comment{up to this block (inclusive)}
  	\item \typ{int} \fld{super}
  		\Comment{approximate index of superblock in \fld{parent.blocks}}
	\item[\com] Blocks in internal \nodes\ have the following additional fields
	\begin{itemize}[leftmargin=3mm]
		\item[$\bullet$] \typ{int} \eleft, \eright
  		\Comment{index of last direct subblock in the left and right child}
	\end{itemize}
  	\item[\com] Blocks in leaf \nodes\ have the following additional field
  	\begin{itemize}[leftmargin=3mm]
		\item[$\bullet$] \typ{Object} \fld{element}
  		\Comment{\var{x} for \opa{Enqueue}{x} operation; otherwise \nl}
	\end{itemize}
	\item[\com] Blocks in the root \node\ have the following additional field
	\begin{itemize}[leftmargin=3mm]
		\item[$\bullet$] \typ{int} \size%
  		\Comment{size of queue after performing all operations up }\linebreak
		\mbox{ }\Comment{to the end of this block}
	\end{itemize}
\end{itemize}




\end{algorithmic}
\caption{Objects used in the \ordering\ tree data structure.\label{object-fields}\Eric{Fix indentation}}
\end{figure}

Each \node\ of the \ordering\ tree has an infinite array called \fld{blocks}.
To simplify the code, \fld{blocks}[0] is initialized with an empty \block\ $B_0$, 
where $E(B_0)$ and $D(B_0)$ are empty sequences.
Each \node's \fld{head} index  stores the position in the \fld{blocks} array to be used
for the next attempt to append a \block.

If a block contained an explicit representation of its sequences of enqueues and dequeues,
it would take $\Omega(p)$ time to construct a block, which would be too slow for task \ref{construct}.
Instead, the block stores an implicit representation of the sequences.
We now explain how we designed the fields for this implicit representation. 
Refer to Figure \ref{implicit} for an example showing how the tree in Figure \ref{orderingtree} is actually represented, and Figure \ref{object-fields} for the definitions of the fields of \blocks\ and \nodes.

A block in a leaf represents a single enqueue or dequeue.  The block's \fld{element} field stores the value
enqueued if the operation is an enqueue, or \nl\ if the operation is a dequeue.

\here{If space permits, we might want to add some examples in the following paragraphs that refer back to Figure \ref{implicit}.}

Each block in an internal \node\ $\var{v}$ has fields \eleft\ and \eright\ that store the indices of the block's last direct subblock in $\var{v}$'s left and right child.  
Thus, the direct subblocks of $\var{v}.\fld{blocks}[b]$ are
\begin{eqnarray}\label{defsubblock}
\var{v}.\fld{left.blocks}[\var{v}.\fld{blocks}[b\!-\!1].\fld{end\sub{left}}\!+\!1..\var{v}.\fld{blocks}[b].\fld{end\sub{left}}] \mbox{ and}\nonumber\\
\var{v}.\fld{right.blocks}[\var{v}.\fld{blocks}[b\!-\!1].\fld{end\sub{right}}\!+\!1..\var{v}.\fld{blocks}[b].\fld{end\sub{right}}].\!
\end{eqnarray}
The \eleft\ and \eright\ fields allow us to navigate to a block's direct subblocks.
Blocks also store some prefix sums:
$\var{v}.\fld{blocks}[b]$ has two fields \fld{sum\sub{enq}} and \fld{sum\sub{deq}}
that store the total numbers of enqueues and dequeues in $\var{v}.\fld{blocks}[1..\var{b}]$.
We use these to search for a particular operation.
For example, consider finding the $r$th enqueue $E_r$ in the linearization.
A binary search for $r$ on \fld{sum\sub{enq}} fields of the root's blocks 
finds the block  containing $E_r$.
If we know a block $B$ in a \node\ $\var{v}$ contains $E_r$,
we can use the \fld{sum\sub{enq}} field again to determine which child of $\var{v}$ contains $E_r$
and then do a binary search
among the direct subblocks of $B$ in that child.
Thus, we work our way down the tree until we find the leaf block that  stores 
$E_r$ explicitly.
We shall show that the binary search in the root can be done in $O(\log p + \log q)$ steps,
and the binary search within each other \node\ along the path to a leaf takes $O(\log p)$ steps,
for a total of $O(\log^2 p + \log q)$ steps for task \ref{findenqueue}.

A block is called the \emph{superblock} of all of its direct subblocks.
To facilitate task \ref{findinroot}, each block $B$ has a field \fld{super} that contains
the (approximate) index of its superblock in the parent \node's \fld{blocks} array (it may differ from the true index by 1).
This allows a process to determine the true location of the superblock by checking the \eleft\ or \eright\ values of just two \blocks\ in the parent \node.
Thus, starting from an operation in a leaf's block, one can use these indices to track the 
operation  up the path to the root, and determine the operation's location in a root block
in $O(\log p)$ time.

Now consider task \ref{findrank}.
To determine whether the queue is empty when a dequeue occurs,
each block in the root has a \fld{size} field storing the number of elements
in the queue after all operations in the linearization up to that block (inclusive) 
have been done.
We can  determine which dequeues in a block $B_d$ in the root are null dequeues using
$B_{d-1}.\fld{size}$, which is the size of the queue just before $B_d$'s operations, and the number of enqueues and dequeues in $B_d$.
Moreover, the total number of non-null dequeues in blocks $B_1, \ldots, B_{d-1}$ 
is $B_{d-1}.\fld{sum\sub{enq}}-B_{d-1}.\fld{size}$.
We can use this information to determine the
rank of a non-null dequeue in $B_d$ among all non-null dequeues in the linearization, which is the rank  (among all enqueues) of the enqueue
whose value the dequeue should return.

Having defined the fields required for tasks \ref{findinroot}, \ref{findrank} and \ref{findenqueue},
we can easily see how to construct a new block $B$ during a \op{Refresh} in $O(1)$ time.
A \op{Refresh} on \node\ $\var{v}$ reads the values $h_{\ell}$ and~$h_{r}$ of the \fld{head} fields of $\var{v}$'s children and stores 
$h_{\ell}-1$ and $h_{r}-1$ in $B.\eleft$ and $B.\eright$.
Then, we can compute 
\begin{eqnarray*}
B.\fld{sum\sub{enq}}&=&\var{v}.\fld{left}.\fld{blocks}[B.\eleft].\fld{sum\sub{enq}} \\
&&+ \var{v}.\fld{right}.\fld{blocks}[B.\eright].\fld{sum\sub{enq}}.
\end{eqnarray*}
For a block $B$ in the root, $B.\fld{size}$ is computed using the \fld{size} field of the previous block $B'$ and
the number of enqueues and dequeues in~$B$:
\begin{eqnarray*}
B.\fld{size} &=& \max(0, B'.\fld{size} + (B.\fld{sum\sub{enq}}-B'.\fld{sum\sub{enq}})\\ 
			&& \hspace*{19mm} - (B.\fld{sum\sub{deq}} - B'.\fld{sum\sub{deq}})).	
\end{eqnarray*}

The only remaining field is $B.\fld{super}$.  When the block 
$B$ is created for a \node\ $\var{v}$, we do not yet know where its
superblock will eventually be installed in \var{v}'s parent.
So, we leave $B.\fld{super}$ blank.  
Soon after $B$ is installed,
some process will
set \var{B}.\fld{super} to a value read from the \fld{head} field of \var{v}'s parent.
We shall show that this happens soon enough that $B.\fld{super}$ can differ from the true index of $B'$
by at most 1.

\subsection{Details of the Implementation}

We now discuss the queue implementation in more detail.  Pseudocode is provided in Figure \ref{pseudocode1}.

\renewcommand{\algorithmiccomment}[1]{\hfill\eqparbox{COMMENTDOUBLE}{\com\ #1}}

\here{For consistency of notation, perhaps change $n$ to $B$ in code for Enqueue, Dequeue.  Maybe also new to B for CreateBlock, Refresh.}

\begin{figure*}
\begin{minipage}[t]{0.405\textwidth}
\begin{algorithmic}[1]
\setcounter{ALG@line}{0}

\Function{void}{Enqueue}{\typ{Object} \var{e}} 
    \State \hangbox{let \var{B} be a new \typ{Block} with \fld{element} \assign\ \var{e},\\
		$\fld{sum\sub{enq}} \assign\ \var{leaf.}\fld{blocks}[\var{leaf.}\head-1].\fld{sum\sub{enq}}+1$,\\
		$\fld{sum\sub{deq}} \assign\ \var{leaf.}\fld{blocks}[\var{leaf.}\head-1].\fld{sum\sub{deq}}$}\label{enqNew}
    \State \Call{Append}{\var{B}}
\EndFunction{Enqueue}

\spac

\Function{Object}{Dequeue()}{} 
    \State \hangbox{let \var{B} be a new \typ{Block} with \fld{element} \assign\ \nl,\\
	    $\fld{sum\sub{enq}} \assign\ \var{leaf.}\fld{blocks}[\var{leaf.}\head-1].\fld{sum\sub{enq}}$,\\
	    $\fld{sum\sub{deq}} \assign\ \var{leaf.}\fld{blocks}[\var{leaf.}\head-1].\fld{sum\sub{deq}}+1$}\label{deqNew}
    \State \Call{Append}{\var{B}}
    \State $\langle \var{b}, \var{i}\rangle$ \assign\ \Call{IndexDequeue}{\var{leaf}, $\var{leaf.}\head-1$, $1$}\label{invokeIndexDequeue}
    \State \Return{ \Call{FindResponse}{\var{b, i}}}\label{deqRest}
\EndFunction{Dequeue}

\spac

\Function{void}{Append}{\typ{Block} \var{B}} \com\ append block to leaf and propagate to root
    \State \var{leaf.}\fld{blocks}[\var{leaf.}\head] \assign\ \var{B}\label{appendLeaf}
    \State $\var{leaf.}\head\ \assign\ \var{leaf.}\head+1$ \label{appendEnd} 
    \State \Call{Propagate}{\var{leaf.}\fld{parent}} 
\EndFunction{Append}

\spac

\Function{void}{Propagate}{\typ{Node} \var{v}} \com\ propagate blocks from \var{v}'s children to root
    \If{\bf{not} \Call{Refresh}{\var{v}}} \label{firstRefresh}  \hfill \com\ double refresh
        \State \Call{Refresh}{\var{v}} \label{secondRefresh}
    \EndIf
    \If{$\var{v} \neq \var{root}$} \hfill \com\ recurse up tree
        \State \Call{Propagate}{\var{v}.\fld{parent}}
    \EndIf
\EndFunction{Propagate}

\spac

\Function{boolean}{Refresh}{\typ{Node} \var{v}} \com\ try to append a new block to \var{v}.\fld{blocks}
    \State \var{h} \assign\ \var{v}.\head \label{readHead}
    \ForEach{\fld{dir} {\keywordfont{in}} $\{\fld{left, right}\}$} \label{startHelpChild1}
        \State \var{childHead} \assign\ \var{v}.\fld{dir}.\head \label{readChildHead}
        \If{\var{v}.\fld{dir.blocks}[\var{childHead}] $\neq$ \nl} \label{ifHeadnotNull}
            \State \Call{Advance}{\var{v}.\fld{dir}, \var{childHead}} \label{helpAdvance}
        \EndIf
    \EndFor \label{endHelpChild1}
    \State \var{new} \assign\ \Call{CreateBlock}{\var{v, h}} \label{invokeCreateBlock}
    \If{\var{new = \nl}} \Return{\tr} \label{addOP} 
	\Else
	    \State \var{result} \assign\ \Call{CAS}{\var{v}.\fld{blocks}[\var{h}], \nl, \var{new}} \label{cas}
    	\State \Call{Advance}{\var{v, h}}\label{advance}
    	\State \Return{ \var{result}}
	\EndIf
\EndFunction{Refresh}

\spac

\Function{Block}{CreateBlock}{\typ{Node} \var{v}, \typ{int} \var{i}} 
    \State \linecomment{create new block for a \op{Refresh} to install in \var{v}.\fld{blocks}[\var{i}]}
    \State let \var{new} be a new \typ{Block} \label{initNewBlock}
    \State \var{new}.\eleft \assign\ $\var{v}.\fld{left}.\head - 1$\label{createEndLeft}
    \State \var{new}.\eright \assign\ $\var{v}.\fld{right}.\head - 1$\label{createEndRight}
	\State \hangbox{\var{new}.\fld{sum\sub{enq}} \assign\ \var{v}.\fld{left.blocks}[\var{new}.\eleft].\fld{sum\sub{enq}} + \\
			\hspace*{11mm}\var{v}.\fld{right.blocks}[\var{new}.\eright].\fld{sum\sub{enq}}}\label{createSumEnq}
	\State \hangbox{\var{new}.\fld{sum\sub{deq}} \assign\ \var{v}.\fld{left.blocks}[\var{new}.\eleft].\fld{sum\sub{deq}} + \\
			\hspace*{11mm}\var{v}.\fld{right.blocks}[\var{new}.\eright].\fld{sum\sub{deq}}}\label{createSumDeq}
    \State \var{num\sub{enq}} \assign\ $\var{new}.\fld{sum\sub{enq}} - \var{v}.\fld{blocks}[\var{i}-1].\fld{sum\sub{enq}}$\label{computeNumEnq}
    \State \var{num\sub{deq}} \assign\ $\var{new}.\fld{sum\sub{deq}} - \var{v}.\fld{blocks}[\var{i}-1].\fld{sum\sub{deq}}$
    \If{$\var{v} = \var{root}$}
        \State \hangbox{\var{new}.\fld{size} \assign\ max(0, $\var{v}.\fld{blocks}[\var{i}-1].\size\ + \var{num\sub{enq}} - \var{num\sub{deq}}$)}\label{computeLength}
    \EndIf
    \If{$\var{num\sub{enq}} + \var{num\sub{deq}} = 0$}\label{testEmpty}
        \State \Return \nl \hfill \com\ no blocks need to be propagated to \var{v}
    \Else
        \State \Return \var{new}
    \EndIf
\EndFunction{CreateBlock}

\end{algorithmic}
\end{minipage}
\begin{minipage}[t]{0.585\textwidth}

\begin{algorithmic}[1]
\setcounter{ALG@line}{57}
\Function{void}{Advance}{\typ{Node} \var{v}, \typ{int} \var{h}} \com\ set \var{v}.\fld{blocks}[\var{h}].\fld{super} and increment \var{v}.\fld{head} from \var{h} to $\var{h}+1$
    \If{$\var{v}\neq \var{root}$}
	    \State \var{h\sub{p}} \assign\ \var{v}.\fld{parent}.\head \label{readParentHead}
    	\State \Call{CAS}{\var{v}.\fld{blocks}[\var{h}].\fld{super}, \nl, \var{h\sub{p}}} \label{setSuper1}
	\EndIf
    \State \Call{CAS}{\var{v}.\head, \var{h}, \var{h}+1} \label{incrementHead}
\EndFunction{Advance}

\spac

\Function{$\langle\typ{int}, \typ{int}\rangle$}{IndexDequeue}{\typ{Node} \var{v}, \typ{int} \var{b}, \typ{int} \var{i}} \com\ return $\langle\var{b}', \var{i}'\rangle$ such that \var{i}th dequeue in
    \State \linecomment{$D(\var{v}.\fld{blocks}[\var{b}])$ is $(\var{i}')$th dequeue of $D(\var{root}.\fld{blocks}[\var{b}'])$}
    \State \linecomment{Precondition: $\var{v}.\fld{blocks}[\var{b}]$ is not \nl, was propagated to root, and contains at least}
    \State \linecomment{\var{i} dequeues}
    \If{$\var{v} = \var{root}$} \Return $\langle\var{b, i}\rangle$ \label{indexBaseCase}
    \Else
	    \State \fld{dir} \assign\ (\var{v}.\fld{parent.left} = \var{v} ? \fld{left} : \fld{right}) 
    	\State \var{sup} \assign\ \var{v}.\fld{blocks}[\var{b}].\fld{super}\label{idsup1}
	    \If{$\var{b} > \var{v}.\fld{parent.blocks}[\var{sup}].\fld{end\sub{dir}}$} \var{sup} \assign\ $\var{sup}+1$\label{supertest}\label{idsup2}
	    \EndIf\label{idsup3}
	    \State \linecomment{compute index \var{i} of dequeue in superblock}
	    \State \hangbox{\var{i} += $\var{v}.\fld{blocks}[\var{b}-1].\fld{sum\sub{deq}} -$ 
	    		$\var{v}.\fld{blocks}[\var{v}.\fld{parent.blocks}[\var{sup}-1].\edir].\fld{sum\sub{deq}}$}\label{computeISuperStart}
        \If{$\fld{dir} = \fld{right}$} 
        	\State \hangbox{\var{i} += $\var{v}.\fld{blocks}[\var{v}.\fld{parent.blocks}[\var{sup}].\eleft].\fld{sum\sub{deq}} - \mbox{ }$\\
					$\var{v}.\fld{blocks}[\var{v}.\fld{parent.blocks}[\var{sup}-1].\eleft].\fld{sum\sub{deq}}$}\label{considerLeftBeforeRight}
        \EndIf \label{computeISuperEnd}
        \State \Return\Call{IndexDequeue}{\var{v}.\fld{parent}, \var{sup}, \var{i}}
    \EndIf
\EndFunction{IndexDequeue}

\spac

\Function{element}{FindResponse}{\typ{int} \var{b}, \typ{int} \var{i}} \com\ find response to \var{i}th dequeue in $D(\var{root}.\fld{blocks}[\var{b}])$
    \State \linecomment{Precondition:  $1\leq i\leq |D(\var{root}.\fld{blocks}[\var{b}])|$}
    \State \hangbox{\var{num\sub{enq}} \assign\ $\var{root}.\fld{blocks}[\var{b}].\fld{sum\sub{enq}} - \var{root}.\fld{blocks}[\var{b}-1].\fld{sum\sub{enq}}$}\label{FRNum}
    \If{$\var{root}.\fld{blocks}[\var{b}-1].\size + \var{num\sub{enq}} < \var{i}$}\label{checkEmpty}
        \State \Return \nl \hfill \com\ queue is empty when dequeue occurs\label{returnNull}
    \Else \ \linecomment{response is the \var{e}th enqueue in the root}
        \State \var{e} \assign\ \var{i} + \var{root}.\fld{blocks}[\var{b}-1].\fld{sum\sub{enq}} - 
			\var{root}.\fld{blocks}[\var{b}-1].\size\label{computeE}
		\State \linecomment{compute enqueue's block using binary search}
		\State find min $b_e\leq \var{b}$ with $\var{root}.\fld{blocks}[b_e].\fld{sum\sub{enq}} \geq \var{e}$\label{FRb}
		\State \linecomment{find rank of enqueue within its block}
		\State $i_e \assign\ \var{e} - \var{root}.\fld{blocks}[b_e-1].\fld{sum\sub{enq}}$\label{FRi}
        \State \Return \Call{GetEnqueue}{\var{root}, $b_e$, $i_e$}\label{findAnswer}
    \EndIf
\EndFunction{FindResponse}

\spac 

\Function{element}{GetEnqueue}{\typ{Node} \var{v}, \typ{int} \var{b}, \typ{int} \var{i}} \Comment{returns argument of \var{i}th enqueue in $E(\var{v}.\fld{blocks}[\var{b}])$}
    \State \linecomment{Preconditions: $\var{i}\geq 1$ and \var{v}.\fld{blocks}[\var{b}] is non-\nl\ and contains at least \var{i} enqueues}
	
    \If{\var{v} is a leaf node} \Return \var{v}.\fld{blocks}[\var{b}].\fld{element} \label{getBaseCase}
    \Else 
        \State \var{sum\sub{left}} \assign\ \var{v}.\fld{left.blocks}[\var{v}.\fld{blocks}[\var{b}].\eleft].\fld{sum\sub{enq}} 
        \State \linecomment{\var{sum\sub{left}} is the number of enqueues in \var{v}.\fld{blocks}[1..$\var{b}$] from \var{v}'s left child}
        \State \var{prev\sub{left}} \assign\ \var{v}.\fld{left.blocks}[\var{v}.\fld{blocks}[$\var{b}-1$].\eleft].\fld{sum\sub{enq}} 
        \State \linecomment{\var{prev\sub{left}} is the number of enqueues in \var{v}.\fld{blocks}[1..$\var{b}-1$] from \var{v}'s left child}
        \State \var{prev\sub{right}} \assign\ \var{v}.\fld{right.blocks}[\var{v}.\fld{blocks}[$\var{b}-1$].\eright].\fld{sum\sub{enq}} 
        \State \linecomment{\var{prev\sub{right}} is the number of enqueues in \var{v}.\fld{blocks}[1..$\var{b}-1$] from \var{v}'s right child}
        \If{$\var{i} \leq \var{sum\sub{left}} - \var{prev\sub{left}}$} \label{leftOrRight} \Comment{required enqueue is in \var{v}.\fld{left}}
            \State \fld{dir} \assign\ \fld{left}
        \Else \Comment{required enqueue is in \var{v}.\fld{right}}
            \State \fld{dir} \assign\ \fld{right}
            \State $\var{i}\ \assign\ \var{i} - (\var{sum\sub{left}} - \var{prev\sub{left}})$
        \EndIf \label{endChooseDir}
        \State \linecomment{Use binary search to find enqueue's block in \var{v}.\fld{dir} and its rank within block}
        \State \hangbox{find minimum $\var{b}'$ in range [\var{v}.\fld{blocks}[$\var{b}-1$].\edir+1..\var{v}.\fld{blocks}[\var{b}].\edir] such that\\
        	 $\var{v}.\fld{dir.blocks}[\var{b}'].\fld{sum\sub{enq}} \geq \var{i} + \var{prev\sub{dir}}$\label{getChild}}
        \State $\var{i}'$ \assign\ $\var{i} - (\var{v}.\fld{dir.blocks}[\var{b}'-1].\fld{sum\sub{enq}} - \var{prev\sub{dir}})$\label{getChildIndex}
        \State \Return\Call{GetEnqueue}{\var{v}.\fld{dir}, $\var{b}'$, $\var{i}'$} \label{getRecurse}
    \EndIf
\EndFunction{GetEnqueue}

\end{algorithmic}
\end{minipage}
\vspace*{-3mm}
\caption{Queue implementation.\label{pseudocode1}}
\end{figure*}

An \opemph{Enqueue}(\var{e}) appends a \block\ to the process's leaf.
The block has $\fld{element}=\var{e}$ to indicate it represents an \op{Enqueue}(\var{e}) operation.
It suffices to propagate the operation to the root and
then use its position in the linearization for future \op{Dequeue}
operations.

A \opemph{Dequeue} also appends a \block\ to the process's leaf.
The block has $\fld{element}=\nl$ to indicate that it represents a \op{Dequeue} operation.
After propagating the operation to the root, it computes
its position in the root using
\op{IndexDequeue} and then computes its response by calling \op{FindResponse}. 

\opemph{Append}(\var{B}) first adds the block \var{B} to the invoking process's leaf.
The leaf's \fld{head} field stores the first empty slot in the leaf's \fld{blocks} array,
so the \op{Append} writes \var{B} there and increments \fld{head}.
Since \op{Append} writes only to the process's own leaf, there cannot be concurrent updates to a leaf.
\op{Append} then calls \op{Propagate} to ensure the operation represented by \var{B} is propagated to the root.

\opemph{Propagate}(\var{v}) guarantees that any blocks that are in \var{v}'s children when \op{Propagate} is invoked are propagated to the root.
It uses the double \op{Refresh} idea described
above and invokes two \op{Refresh}es on \var{v} in Lines
\ref{firstRefresh} and \ref{secondRefresh}. 
If both fail to add a block to \var{v}, it means some other process has done a successful \op{Refresh}
that propagated blocks that were in \var{v}'s children prior to line \ref{firstRefresh} to \var{v}.
Then, \op{Propagate} recurses to \var{v}.\fld{parent} to continue propagating blocks up to the root.  

A \opemph{Refresh} on node \var{v} creates a block representing the new blocks
in \var{v}'s
children and tries to append it to \var{v}.\fld{blocks}. 
Line \ref{readHead} reads \var{v}.\fld{head} into the local variable \var{h}.
Line \ref{invokeCreateBlock} creates the new block to install in \var{v}.\fld{blocks}[\var{h}].
If line \ref{invokeCreateBlock} returns \nl\ instead of a new block, there were no new blocks in \var{v}'s children to propagate to \var{v}, so \op{Refresh} can return true at line \ref{addOP} and terminate.
Otherwise, the CAS at line \ref{cas} tries to install the new block into \var{v}.\fld{blocks}[\var{h}].
Either this CAS succeeds or some other process has installed a  block in this location.
Either way, line \ref{advance} then calls \opemph{Advance} to advance \var{v}'s head index 
from \var{h} to $\var{h}+1$
and fill in the \fld{super} field of the most recently appended block.
The boolean value returned by \op{Refresh} indicates whether its CAS succeeded.
A \op{Refresh} may pause after a successful CAS before calling \op{Advance} at line \ref{advance},
so other processes help keep \fld{head} up to date by  calling \op{Advance}, 
either at line \ref{helpAdvance} during a \op{Refresh} on \var{v}'s parent or line \ref{advance} during a \op{Refresh} on~\var{v}.

\opemph{CreateBlock}(\var{v, i}) is used
by \op{Refresh} to construct a block to be installed in \var{v}.\fld{blocks}[\var{i}].
The \eleft\ and \eright\ fields store the indices of the last blocks appended to \var{v}'s
children, obtained by reading the \fld{head} index in \var{v}'s children.
Since the \fld{sum\sub{enq}} field should store the number of enqueues in
\var{v}.\fld{blocks}[1..\var{i}] and these enqueues come from \var{v}.\fld{left.blocks}[1..\var{new}.\eleft] and \var{v}.\fld{blocks}[1..\var{new}.\eright], line \ref{createSumEnq} sets
\fld{sum\sub{enq}} to the sum of $\var{v}.\fld{left.blocks}[\var{new}.\eleft].\fld{sum\sub{enq}}$ and $\var{v}.\fld{right}.\fld{blocks}[\var{new}.\eright].\fld{sum\sub{enq}}$.
Line \ref{computeNumEnq} sets \var{num\sub{enq}} to the number of enqueues in the new block by
subtracting  the number of enqueues  in \var{v}.\fld{blocks}[$1..\var{i}-1$] from \var{new}.\fld{sum\sub{enq}}.
The values of \var{new}.\fld{sum\sub{deq}} and \var{num\sub{deq}} are computed similarly.
Then, if \var{new}
is going to be installed in the root, line \ref{computeLength} computes the \fld{size} field, which
represents the number of elements in the queue after the operations in the block are performed.
Finally, if the new block contains no operations, \op{CreateBlock} returns \nl\ to indicate
 there is no need to install it.

Once a dequeue is appended to a block of the process's leaf and propagated to the root,
the \opemph{IndexDequeue} routine finds the dequeue's location in the root.
More precisely, \opa{IndexDequeue}{v, b, i}
computes the block in the root and the rank
within that block  of the \var{i}th dequeue of the block \var{B} stored in \var{v}.\fld{blocks}[\var{b}].
Lines \ref{idsup1}--\ref{idsup3} compute the location of $B$'s superblock in \var{v}'s parent, taking into account the fact that $B.\fld{super}$ may be off by one from the superblock's true index.
The arithmetic in lines \ref{computeISuperStart}--\ref{computeISuperEnd} compute the dequeue's 
rank within the superblock's sequence of dequeues, using  (\ref{defSeqs}).


To compute the response of the $i$th \op{Dequeue} in the $b$th block
of the root, \opemph{FindResponse}(\var{b, i}) determines at line \ref{checkEmpty} if the queue is empty.
If not, line \ref{computeE} computes the rank \var{e} of the
\op{Enqueue} whose argument is the \op{Dequeue}'s response. 
A binary search on the \fld{sum\sub{enq}} fields of \var{root}.\fld{blocks} finds the index $b_e$ of the block that contains 
the \var{e}th enqueue.
Since the enqueue is linearized before the dequeue, $b_e\leq b$.  To find the left end of the range for the binary search for $b_e$, we can first do a doubling search \cite{BY76}, comparing \var{e} to the \fld{sum\sub{enq}} fields at indices $b-1, b-2, b-4, b-8, \ldots$.
Then, \op{GetEnqueue} is used to trace down through the tree to find the required enqueue in a leaf.

\opemph{GetEnqueue}(\var{v, b, i}) returns the argument of the
$i$th enqueue in the $b$th block $B$ of \typ{Node} $\var{v}$. 
It recursively finds the location of the enqueue in each node along the path from $v$ to a leaf, which stores the argument explicitly.
\op{GetEnqueue} first determines which child of \var{v} contains the enqueue, and
then finds the range of blocks within that child that are subblocks of $B$ using information stored
in $B$ and the block that precedes $B$ in $v$.
\op{GetEnqueue} finds the exact subblock containing the enqueue using a binary search on the \fld{sum\sub{enq}}
field (line~\ref{getChild}) and proceeds recursively down the tree.


\section{Proof of Correctness}
\label{sec::correctness}

After proving some basic properties in Section \ref{sec::basicProperties},
we show in Section \ref{sec::propagating} that a double refresh at each node
suffices to propagate an operation to the root.
In Section \ref{sec::tracingCorrect} we show \op{GetEnqueue} and \op{IndexDequeue}
correctly navigate through the  tree.
Finally, we prove linearizability in Section \ref{sec::linearizability}.

\subsection{Basic Properties}
\label{sec::basicProperties}

A \typ{Block} object's fields, except for \fld{super}, are immutable:  they are written only 
when the block is created.
Moreover, only a \op{CAS} at line \ref{setSuper1}  modifies  \fld{super}  
(from \nl\ to a non-\nl\ value), so it is changed only once.
Similarly, only a \op{CAS} at line \ref{cas} modifies an element of a node's \fld{blocks} array 
(from \nl\ to a non-\nl\ value), so blocks are permanently added to nodes.
Only a \op{CAS} at line \ref{incrementHead} can update a node's \head\ field by incrementing it,
which implies the following.

\begin{observation} \label{nonDecreasingHead}
For each node \var{v},  \var{v}.\fld{head} is non-decreasing over time.
\end{observation}

\begin{observation} \label{lem::headInc}
Let $R$ be an instance of \opa{Refresh}{v} whose call to \op{CreateBlock} returns a non-\nl\ block.  When $R$ terminates, \var{v}.\head\ is strictly greater than the value $R$ reads from it at line \ref{readHead}.
\end{observation}
\begin{proof}
After $R$'s \op{CAS} at line \ref{incrementHead}, \var{v}.\head\ is no longer equal to the value \var{h}
read at line \ref{readHead}.  The claim follows from \Cref{nonDecreasingHead}.
\end{proof}

Now we show $\var{v}.\fld{blocks}[\var{v}.\head]$ is either the last non-\nl\ block or the first \nl\ block in node $\var{v}$.

\begin{invariant}\label{lem::headPosition} 
For $0 \leq i < \var{v}.\head$, $\var{v}.\fld{blocks}[i]\neq\nl$.  For $i>\var{v}.\head$, $\var{v}.\fld{blocks}[i]=\nl$.
If $\var{v}\neq \var{root}$,  $\var{v}.\fld{blocks}[i].super \neq \nl$ for $0<i<\var{v}.\head$.
\end{invariant}

\begin{proof}
Initially, $\var{v}.\head=1$, $\var{v}.\fld{blocks}[0]\neq\nl$  and $\var{v}.\fld{blocks}[i]=\nl$ for  $i>0$, so the claims~hold.

Assume the claims hold before a change to $\var{v}.\fld{blocks}$, which can be made only
by a successful \op{CAS} at line \ref{cas}.
The \op{CAS} changes $\var{v}.\fld{blocks}[h]$ from \nl\ to a non-\nl\ value.
Since $\var{v}.\fld{blocks}[h]$ is \nl\ before the CAS, $\var{v}.\head \leq h$ by the hypothesis.
Since $h$ was read from $\var{v}.\head$ earlier at line \ref{readHead}, the current value of 
$\var{v}.\head$ is at least $h$ by \Cref{nonDecreasingHead}.
So, $\var{v}.\head=h$ when the \op{CAS} occurs and a change to $\var{v}.\fld{blocks}[\var{v}.\head]$ preserves the invariant.

Now, assume the claim holds before a change to $\var{v}.\head$, which can only be an increment from $h$ to $h+1$
by a successful \op{CAS} at line \ref{incrementHead} of \op{Advance}.
For the first two claims, it suffices to show that $\var{v}.\fld{blocks}[head] \neq \nl$.
\nf{Advance} is called either at line \ref{helpAdvance} 
after testing that $\var{v}.\fld{blocks}[h]\neq\nl$ at line \ref{ifHeadnotNull},
or at line \ref{advance} after the \op{CAS} at line \ref{cas} ensures $\var{v}.\fld{blocks}[h]\neq\nl$.
For the third claim, observe that prior to incrementing $\var{v}.\head$ to $i+1$ at line \ref{incrementHead},
the \op{CAS} at line \ref{setSuper1} ensures that $\var{v}.\fld{blocks}[i].super\neq \nl$.
\end{proof}

It follows that blocks accessed by the \op{Enqueue}, \op{Dequeue} and \op{CreateBlock} routines are non-\nl.

The following two lemmas show that no operation appears in more than one block of the root.
\begin{lemma} \label{lem::headProgress}
 If $b>0$ and $\var{v}.\fld{blocks}[b] \neq \nl$, then
 \begin{eqnarray*}
 \var{v}.\fld{blocks}[b-1].\fld{end\sub{left}} &\leq& \var{v}.\fld{blocks}[b].\fld{end\sub{left}} \mbox{ and}\\
 \var{v}.\fld{blocks}[b-1].\fld{end\sub{right}} &\leq& \var{v}.\fld{blocks}[b].\fld{end\sub{right}}.
 \end{eqnarray*}
\end{lemma}
\begin{proof}
Let $B$ be the block in $\var{v}.\fld{blocks}[b]$.
Before creating $B$ at line \ref{invokeCreateBlock}, the \op{Refresh} that installed $B$
read $b$ from $\var{v}.\head$ at line \ref{readHead}.
At that time, $\var{v}.\fld{blocks}[b-1]$ contained a block $B'$, by \Cref{lem::headPosition}.
Thus, the \op{CreateBlock}($\var{v},b-1$) that created $B'$ terminated before the \op{CreateBlock}($\var{v},b$) that
created $B$ started.
It follows from \Cref{nonDecreasingHead} that the value that 
line \ref{createEndLeft} of \op{CreateBlock}($\var{v},b-1$) stores in $B'.\fld{end\sub{left}}$   
is less than or equal to the value that line \ref{createEndLeft} of \op{CreateBlock}($\var{v},b$) 
stores in $B.\fld{end\sub{left}}$.
Similarly, the values stored in $B'.\eright$ and $B.\eright$ at line \ref{createEndRight} 
satisfy the claim.
\end{proof}

\begin{lemma} \label{lem::subblocksDistinct}
If $B$ and $B'$ are two blocks in nodes at the same depth, their sets of subblocks are disjoint.
\end{lemma}
\begin{proof}
We prove the lemma by reverse induction on the depth.
If $B$ and $B'$ are in leaves, they have no subblocks, so the claim holds.
Assume the claim holds for nodes at depth $d+1$ and let $B$ and $B'$ be two blocks in nodes at depth $d$.
Consider the direct subblocks of $B$ and $B'$ defined by~(\ref{defsubblock}).
If $B$ and $B'$ are in different nodes at depth $d$, then their direct subblocks are disjoint.
If $B$ and $B'$ are in the same node, it follows from \Cref{lem::headProgress} that their direct subblocks are disjoint.
Either way, their direct subblocks (at depth $d+1$) are disjoint, so the claim follows from the induction hypothesis.
\end{proof}

It follows that each block has at most one superblock.
Moreover, we can now prove each operation is contained in at most one block of each node,
and hence appears at most once in the linearization~$L$.

\here{Might be able to get rid of this corollary and just cite previous lemma instead to save space}
\begin{corollary}\label{lem::noDuplicates}
For  $i\neq j$, $\var{v}.\fld{blocks}[i]$ and $\var{v}.\fld{blocks}[j]$ cannot both contain the same operation.
\end{corollary}
\begin{proof}
A block $B$ contains the operations in $B$'s subblocks in leaves of the tree.
An operation by process $P$ appears in just one block of $P$'s leaf, so
an operation 
cannot be in two different leaf blocks. 
By \Cref{lem::subblocksDistinct}, $\var{v}.\fld{blocks}[i]$ and $\var{v}.\fld{blocks}[j]$ have no common subblocks, so the claim follows.
\end{proof}

%


The accuracy of the values stored in the \fld{sum\sub{enq}} and \fld{sum\sub{deq}} fields
on lines \ref{enqNew}, \ref{deqNew}, \ref{createSumEnq} and \ref{createSumDeq} follows easily
from the definition of subblocks.  See Appendix 
\ref{app::tracingDetails} for a detailed proof of \Cref{lem::sum}.

\begin{restatable}{invariant}{sumRes}
\label{lem::sum}
If $B$ is a block stored in $\var{v}.\fld{blocks}[i]$, then
\begin{eqnarray*}
B.\fld{sum\sub{enq}} &=& | E(\var{v}.\fld{blocks}[0])\cdots E(\var{v}.\fld{blocks}[i]) | \mbox{ and}\\
B.\fld{sum\sub{deq}} &=& | D(\var{v}.\fld{blocks}[0])\cdots D(\var{v}.\fld{blocks}[i]) |.
\end{eqnarray*}
\end{restatable}

This allows us to prove that every block a \op{Refresh} installs contains at least one operation.

\begin{corollary}\label{blockNotEmpty}
If a block $B$ is in $\var{v}.\fld{blocks}[i]$ where $i>0$, then $E(B)$ and $D(B)$ are not both empty.
\end{corollary}
\begin{proof}
The \op{Refresh} that installed $B$ got $B$ as the response to its call to \op{CreateBlock} on line \ref{invokeCreateBlock}.
Thus, at line \ref{testEmpty}, $\var{num\sub{enq}}+\var{num\sub{deq}}\neq 0$.
By \Cref{lem::sum}, $\var{num\sub{enq}} = |E(B)|$ and $\var{num\sub{deq}} = |D(B)|$,
so these sequences cannot both be empty.
\end{proof}

\subsection{Propagating Operations to the Root}
\label{sec::propagating}

In the next two lemmas, we show two \op{Refresh}es suffice to propagate operations from a child to its parent.
We say that node $\var{v}$ \emph{contains} an operation $op$ if some block in $\var{v}.\fld{blocks}$ contains $op$.
Since blocks are permanently added to nodes, if $\var{v}$ contains $op$ at some time, $v$ contains $op$ at all later times too.

\begin{lemma}\label{successfulRefresh}
Let $R$ be a call to \op{Refresh}($\var{v}$) that performs a successful \op{CAS} on line \ref{cas} (or terminates at line \ref{addOP}).
After that CAS (or termination, respectively), $\var{v}$ contains all operations that $\var{v}$'s children contained 
when $R$ executed line~\ref{readHead}.
\end{lemma}
\begin{proof}
Suppose $\var{v}$'s child (without loss of generality, $\var{v}.\fld{left}$) contained an operation $op$ 
when $R$ executed line \ref{readHead}.
Let $i$ be the index such that the block $B=\var{v}.\fld{left.blocks}[i]$ contains $op$.
By \Cref{nonDecreasingHead} and \Cref{lem::headProgress}, the value of $childHead$ that $R$ reads from
$\var{v}.\fld{left.head}$ in line \ref{readChildHead} is at least $i$.
If it is equal to $i$, $R$ calls \op{Advance} at line \ref{helpAdvance}, which ensures that 
$\var{v}.\fld{left.head} > i$.
Then, $R$ calls \op{CreateBlock}($\var{v},h$) in line \ref{invokeCreateBlock}, where $h$ is the value $R$ reads at line \ref{readHead}.
\op{CreateBlock} reads a value greater than $i$ from $\var{v}.\fld{left.head}$ at line \ref{createEndLeft}.
Thus, $new.\eleft \geq i$.  We consider two cases.

Suppose $R$'s call to \op{CreateBlock} returns the new block $B'$ and $R$'s \op{CAS} at line \ref{cas} 
installs $B'$ in $\var{v}.\fld{blocks}$.
Then, $B$ is a subblock of some block in $\var{v}$, since  $B'.\eleft$ is greater than or equal to $B$'s index
$i$ in $\var{v}.\fld{left.blocks}$.
Hence $\var{v}$ contains $op$, as required.

Now suppose $R$'s call to \op{CreateBlock} returns \nl, causing $R$ to terminate at line \ref{addOP}.
Intuitively, since there are no operations in $\var{v}$'s children to promote, $op$ is already in $\var{v}$.
We formalize this intuition.
The value computed at line \ref{createSumEnq} is
\begin{eqnarray*}
\var{num\sub{enq}} \hspace*{-2mm}
&=& \hspace*{-2mm} \var{v}.\fld{left.blocks}[new.\eleft].\fld{sum\sub{enq}} \\
&& + \var{v}.\fld{right.blocks}[new.\eright].\fld{sum\sub{enq}}\\
&&  - \var{v}.\fld{blocks}[h-1].\fld{sum\sub{enq}} \\
&=& \hspace*{-2mm}\var{v}.\fld{left.blocks}[new.\eleft].\fld{sum\sub{enq}}\\
&& + \var{v}.\fld{right.blocks}[new.\eright].\fld{sum\sub{enq}} \\
&&- \var{v}.\fld{left.blocks}[\var{v}.\fld{blocks}[h-1].\eleft].\fld{sum\sub{enq}} \\
&& - \var{v}.\fld{right.blocks}[\var{v}.\fld{blocks}[h-1].\eright].\fld{sum\sub{enq}}.
\end{eqnarray*}
It follows from \Cref{lem::sum} that $num\sub{enq}$ is the total  number of enqueues  in 
$\var{v}.\fld{left.blocks}[\var{v}.\fld{blocks}[h-1].\eleft+1..new.\eleft]$ and
$\var{v}.\fld{right.blocks}[\var{v}.\fld{blocks}[h-1].\eright+1..new.\eright]$.
Similarly, $num\sub{deq}$ is the total number of dequeues contained in these blocks.
Since $num\sub{enq}+num\sub{deq}=0$ at line \ref{testEmpty},
these blocks contain no operations.
By \Cref{blockNotEmpty}, this means the ranges of blocks are empty, so that $\var{v}.\fld{blocks}[h-1].\eleft \geq \var{new}.\eleft \geq i$.
Hence, $B$ is already a subblock of some block in $\var{v}$, so $\var{v}$ contains $op$.
\end{proof}

We now show  a double \op{Refresh} propagates blocks as required.

\begin{lemma}\label{lem::doubleRefresh}
Consider two consecutive terminating calls $R_1$, $R_2$ to \op{Refresh}($\var{v}$) by the same process.
All operations contained $\var{v}$'s children when $R_1$ begins
are contained in $\var{v}$ when $R_2$ terminates.
\end{lemma}
\begin{proof}
If either $R_1$ or $R_2$ performs a successful \op{CAS} at line \ref{cas} or terminates at line \ref{addOP}, the claim follows
from \Cref{successfulRefresh}.
So suppose both $R_1$ and $R_2$ perform a failed \op{CAS} at line \ref{cas}.
Let $h_1$ and $h_2$ be the values $R_1$ and $R_2$ read from $\var{v}.\head$ at line \ref{readHead}.
By \Cref{lem::headInc}, $h_2>h_1$.
By \Cref{lem::headProgress}, $\var{v}.blocks[h_2]=\nl$ when $R_1$ executes line \ref{readHead}.
Since $R_2$ fails its \op{CAS} on $\var{v}.blocks[h_2]$, some other \op{Refresh} $R_3$ must have done
a successful \op{CAS} on $\var{v}.blocks[h_2]$ before $R_2$'s \op{CAS}.
$R_3$ must have executed line \ref{readHead} after $R_1$, since $R_3$ read the value $h_2$ from $\var{v}.\head$ and the value of $\var{v}.\head$ is non-decreasing, by \Cref{nonDecreasingHead}.
Thus, all operations contained in $\var{v}$'s children when $R_1$ begins
are also contained in $\var{v}$'s children when $R_3$ later executes line \ref{readHead}.
By \Cref{successfulRefresh}, these operations are contained in $\var{v}$ when $R_3$ performs its successful \op{CAS},
which is before $R_2$'s failed \op{CAS}.
\end{proof}

\begin{lemma} \label{lem::appendExactlyOnce}
When an \op{Append}($B$) terminates, $B$'s operation is contained in exactly one block in each node along the path from the process's leaf to the root.
\end{lemma}
\begin{proof}
\op{Append} adds $B$ to the process's leaf and calls \op{Propagate}, which
does a double \op{Refresh}~on each internal node on the path $P$ from the leaf to the root.
By \Cref{lem::doubleRefresh}, this ensures a block in each node on $P$ contains $B$'s operation.
There is at most one such block in each node, by \Cref{lem::noDuplicates}.
\end{proof}

\subsection{Correctness of \opemph{GetEnqueue} and \opemph{IndexDequeue}}
\label{sec::tracingCorrect}

See Appendix \ref{app::tracingDetails} for detailed proofs for this section.

We first show the \fld{super} field is accurate, since \op{IndexDequeue} uses it to trace superblocks up the tree.  This is proved by showing that the \fld{super} field of a block $B$ in node $\var{v}$ is read from  
\var{v}.\fld{parent}'s \var{head} field close to the time that $B$'s superblock $B_s$ is installed in the parent node.
On one hand, $B.super$ is written before $B_s$ is installed:
\op{Advance} writes $B.super$ before advancing $\var{v}.\head$  past $B$'s index, which must happen before the \op{CreateBlock} that creates
$B_s$ gets the value of $B_s.\eleft$ or $B_s.\eright$.
On the other hand, $B.\fld{super}$ cannot be written too long before $B_s$ is installed:
$B.\fld{super}$ is written after $B$ is installed, and \Cref{successfulRefresh} ensures that $B$ is propagated to the
parent soon after.

\begin{restatable}{lemma}{superRelationRes}
\label{superRelation}
Let $B=\var{v}.\var{blocks}[b]$.
  If $\var{v}.\fld{parent.blocks}[s]$ is the superblock of $B$ then $s-1\leq B.\fld{super}\leq s$.
\end{restatable}

To show  \op{GetEnqueue} and \op{IndexDequeue} work correctly, we  just  check that they correctly compute the index of the required block
and the operation's rank within the block.  
For \op{IndexDequeue}, we use \Cref{superRelation} each time \op{IndexDequeue} goes one step up the tree.
\here{mention somewhere why precondition of IndexDequeue is true? also check that block of parent indexed by sup in that routine is non-null}

\begin{restatable}{lemma}{indexDequeueRes}
\label{lem::indexDequeue}
If $\var{v}.\fld{blocks}[b]$ has been propagated to the root and $1\leq i\leq |D(\var{v}.\fld{blocks}[b])|$, 
 then \op{IndexDequeue}($\var{v}, b, i$) returns $\langle b',i' \rangle$ such that the \var{i}th dequeue in $D(\var{v}.\fld{blocks}[\var{b}])$ is the $(i')$th dequeue of $D(\var{root}.\fld{blocks}[b'])$.
\end{restatable}

\begin{restatable}{lemma}{getEnqRes}
\label{lem::get}
If $1\leq i\leq |E(\var{v}.\fld{blocks}[b])|$ then \op{getEnqueue}($\var{v},b,i$) returns the argument of the $i$th enqueue in $E(\var{v}.\fld{blocks}[b])$.
\end{restatable}

\subsection{Linearizability}
\label{sec::linearizability}

We show that the linearization ordering $L$ defined in 
(\ref{linearization}) is a legal permutation of a subset of the operations in 
the execution, i.e., that it includes all operations that terminate and 
if one operation $op_1$ terminates before another operation $op_2$ begins, then $op_2$ does not precede $op_1$ in $L$.  We also show the result each dequeue returns is the same as in the sequential execution  $L$.

\begin{lemma} \label{linearSat}
$L$ is a legal linearization ordering.
\end{lemma}
\begin{proof}
By \Cref{lem::noDuplicates}, $L$ is a permutation of a subset of the operations in the execution.
By \Cref{lem::appendExactlyOnce}, each terminating operation is propagated to the root before it terminates,
so it appears in $L$.
Also, if $op_{1}$ terminates before $op_{2}$ begins, then $op_{1}$ 
is propagated to the root before $op_2$ begins, so $op_1$ appears before $op_2$ in $L$.
\end{proof}

A simple proof (in Appendix \ref{app::tracingDetails}) shows that \fld{size} fields are computed correctly.
\begin{restatable}{lemma}{sizeCorrectRes}
\label{sizeCorrectness}
If the operations of $\var{root}.\fld{blocks}[0..b]$ are applied sequentially in the order of~$L$ on an initially empty queue, the resulting queue has $\var{root}.\fld{blocks}[b].\size$ elements.  
\end{restatable}

Next, we show operations return the same response as they would in the sequential execution $L$.

\begin{lemma}\label{linearCorrect}
Each terminating dequeue returns the response it would in the sequential execution $L$.
\end{lemma}
\begin{proof}
If a dequeue $D$ terminates, it is contained in some block in the root, by \Cref{lem::appendExactlyOnce}.
By \Cref{lem::indexDequeue}, $D$'s call to \op{IndexDequeue} on line \ref{invokeIndexDequeue}
returns a pair $\langle b,i\rangle$ such that $D$ is the $i$th dequeue in the block 
$B=\var{root}.\fld{blocks}[b]$.
$D$ then calls \op{FindResponse}($b,i$) on line \ref{deqRest}.
By \Cref{sizeCorrectness}, the queue contains $\var{root}.\fld{blocks}[b-1].\fld{size}$ elements
after the operations in $\var{root}.\fld{blocks}[1..b-1]$ are performed sequentially 
in the order given by $L$.
By \Cref{lem::sum}, the value of \var{num\sub{enq}} computed on line \ref{FRNum}
is the number of enqueues in $B$.
Since the enqueues in block $B$ precede the dequeues,
the queue is empty when the $i$th dequeue of $B$ occurs if 
$\var{root}.\fld{blocks}[b-1].\fld{size} + \var{num\sub{enq}} < i$.
So $D$ returns \nl\ on line \ref{returnNull} if and only if it would do so in the sequential
execution $L$.
Otherwise, the size of the queue after doing the operations in $\var{root}.\fld{blocks}[0..b-1]$
in the sequential execution $L$ is $\var{root}.\fld{blocks}[b-1].\fld{sum\sub{enq}}$ minus
the number of non-\nl\ dequeues in that prefix of $L$.
Hence, line \ref{computeE} sets $e$ to the rank of $D$ among all the non-\nl\ dequeues in $L$.
Thus, in the sequential execution~$L$, $D$ returns the value enqueued by the $e$th enqueue in $L$.
By \Cref{lem::sum}, this enqueue is the $i_e$th enqueue 
in $E(\var{root}.\fld{blocks}[b_e])$, where
$b_e$ and $i_e$ are the values $D$ computes on line \ref{FRb} and \ref{FRi}.
By \Cref{lem::get}, the call to \op{GetEnqueue} returns the argument of the required enqueue.
\end{proof}

Combining \Cref{linearSat} and \Cref{linearCorrect} provides our main result.

\begin{mytheorem}
The queue implementation is linearizable.
\end{mytheorem}

\here{Should there be a lemma somewhere in this section that explicitly shows that each
object we dereference is non-null?}


\section{Analysis}
We now analyze the number of steps and the number of \op{CAS} instructions performed by operations.

\begin{proposition}\label{casbound}
Each \op{Enqueue} or \op{Dequeue} operation performs $O(\log p)$ \op{CAS} instructions.
\end{proposition}
\begin{proof}
An operation invokes \op{Refresh} at most twice at each of the $\ceil{\log_2 p}$ levels of the tree.
A \op{Refresh} does at most 5 \op{CAS} steps: one in line \ref{cas} and two during each \op{Advance} in line \ref{helpAdvance} or~\ref{advance}.
\end{proof}

\begin{lemma}\label{dSearchTime}
The search that \op{FindResponse}$(b,i)$ does at line \ref{FRb} to find the index $b_e$ of the block in the root containing the $e$th enqueue takes $O(\log (\var{root}.\fld{blocks}[b_e].\size + \var{root}.\fld{blocks}[b-1].\size))$ steps.
\end{lemma}
\begin{proof}
Let the blocks in the root be $B_1, \ldots, B_\ell$.
The doubling search for $b_e$ takes $O(\log (b-b_e))$ steps,
so we prove $b - b_e \leq 2 \cdot B_{b_e}.\size + B_{b-1}.\size + 1$.
If $b \leq b_e+1$, then this is trivial, so assume  $b>b_e+1$.
As shown in \Cref{linearCorrect}, the dequeue that calls \op{FindResponse} is in $B_b$ and is supposed to return an enqueue in $B_{b_e}$.
Thus, there can be at most $B_{b_e}.\size$ dequeues in 
$D(B_{b_e+1}) \cdots D(B_{b-1})$; otherwise in the sequential execution $L$,
all elements enqueued before the end of
$E(B_{b_e})$ would be dequeued before $D(B_b)$. 
Furthermore, by \Cref{sizeCorrectness}, the size of the queue  after the prefix of $L$ corresponding to 
$B_1,\ldots,B_{b-1}$  is 
$B_{b-1}.\size \geq B_{b_e}.size + |E(B_{b_e+1})\cdots E(B_{b-1})| - |D(B_{b_e+1})\cdots D(B_{b-1})|$.
Thus, $|E(B_{b_e+1})\cdots E(B_{b-1})| \leq B_{b-1}.\size + |D(B_{b_e+1})\cdots D(B_{b-1})| \leq B_{b-1}.\size + B_{b_e}.\size$.
So, the total number of operations in $B_{b_e+1}, \ldots, B_{b-1}$ is at most
$B_{b-1}.\size + 2\cdot B_{b_e}.\size$.
Each of these $b-1-b_e$ blocks contains at least one operation, by \Cref{blockNotEmpty}.
So, $b-1-b_e \leq B_{b-1}.\size + 2\cdot B_{b_e}.\size$.
\end{proof}

The following lemma helps bound the time for \op{GetEnqueue}.
\begin{lemma}\label{blockSize}
Each block $B$ in each node contains at most one operation of each process.
If $c$ is the execution's maximum point contention, $B$ has at most $c$ direct~subblocks.
\end{lemma}
\begin{proof}
Suppose $B$ contains an operation of process $p$.
Let $op$ be the earliest operation by $p$ contained in $B$.
When $op$ terminates, $op$ is contained in $B$ by \Cref{lem::appendExactlyOnce}.
Thus, $B$ cannot contain any later operations by $p$, since $B$ is created before
those operations are invoked.

Let $t$ be the earliest termination of any operation contained in $B$.
By \Cref{lem::appendExactlyOnce}, $B$ is created before $t$, so all operations contained in $B$
are invoked before $t$.  Thus, all are  running concurrently at $t$, so $B$ contains at most $c$ operations.
By definition, the direct subblocks of $B$ contain these $c$ operations, and each operation is contained
in exactly one of these subblocks, by Lemma \ref{lem::subblocksDistinct}.
By \Cref{blockNotEmpty}, each direct subblock of $B$ contains at least one operation,
so $B$ has at most $c$ direct subblocks.
\end{proof}

We now bound step complexity in terms of the number of processes $p$, the maximum contention $c\leq p$, and the size of the queue. 

\begin{mytheorem}\label{enqDeqTime}
Each \op{Enqueue} and null Dequeue takes $O(\log p)$ steps 
and each non-null \op{Dequeue} takes
$O(\log p\log c + \log q_e+ \log q_d)$ steps,
where $q_d$ is the size of the queue when the \op{Dequeue} is linearized and 
$q_e$ is the size of the queue when the \op{Enqueue} of the value returned is linearized.
\end{mytheorem}
\begin{proof}
An \op{Enqueue} or null \op{Dequeue} creates a block, appends it to the process's 
leaf and propagates it to the root.  The \op{Propagate}  does $O(1)$ steps 
at each node on the path from the process's leaf to the root.
A null \op{Dequeue} additionally calls \op{IndexDequeue}, which also does $O(1)$ steps
at each node on this path. 
So, the total number of steps for either type of operation is $O(\log p)$.

A non-null \op{Dequeue} must also search at line \ref{FRb} and call \op{GetEnqueue}
at line \ref{findAnswer}.
By \Cref{dSearchTime}, the doubling search takes $O(\log(q_e+q_d+p))$ steps,
since the size of the queue can change by at most $p$ within one block (by \Cref{blockSize}).
\op{GetEnqueue} does a binary search within each node on a path from the root to a leaf.
Each node $\var{v}$'s search is within the subblocks of one block in $\var{v}$'s parent.
By \Cref{blockSize}, each such search takes $O(\log c)$ steps, for a total of $O(\log p\log c)$ steps.
\end{proof}

\begin{corollary}
The queue implementation is wait-free.
\end{corollary}


\section{Bounding Space Usage}
\label{reducing}

Operations remain in the \fld{blocks} arrays forever. 
This uses space proportional to the number of enqueues that have been invoked.
Now, we modify the implementation to remove blocks that are no longer needed, so that space usage is
polynomial in $p$ and $q$, while (amortized) step complexity is still polylogarithmic.  For lack of space, details are in Appendix~\ref{reducing-details}.
We replace the \fld{blocks} array in each node by a red-black tree (RBT)
that stores the blocks.
Each block has an additional \fld{index} field
that represents its position within the original \fld{blocks} array, and
blocks in a RBT are sorted by \fld{index}.
The attempt to install a new block in $\fld{blocks}[i]$  on line \ref{cas}
is replaced by an attempt to insert a new block with index $i$ into the RBT.
Accessing the block in $\fld{blocks}[i]$ is 
replaced by searching the RBT for the  index $i$.
The binary searches for a block in line \ref{FRb} and \ref{getChild} can simply search the RBT
using the \fld{sum\sub{enq}} field, since the RBT is also sorted with respect to this field, by \Cref{lem::sum}.
 
Known lock-free search trees have step complexity that includes a term linear in $p$ \cite{EFHR14,Ko20}.  
However, we do not require all the standard search tree operations.
Instead of a standard insertion, we allow a \op{Refresh}'s insertion to fail if another
concurrent \op{Refresh} succeeds in inserting a block, just as the \op{CAS} on line \ref{cas}
can fail if a concurrent \op{Refresh} does a successful \op{CAS}.
Moreover, the insertion should succeed only if the newly inserted block has a larger index than any other block in the RBT.
Thus, we can use a particularly simple concurrent RBT implementation.
A sequential RBT can be made persistent using the classic node-copying technique of 
Driscoll et al.~\cite{DSST89}:  all RBT nodes are immutable, and operations on the 
RBT make a new copy of each RBT node $x$ that must be modified, as well
as each RBT node along the path from the RBT's root to~$x$.
The RBT reachable from the new copy of the root is the result of applying the RBT operation.
This only adds a constant factor to the running time of any routine designed for a (sequential) RBT.
Once a process has performed an update to the RBT representing the blocks of a node 
$\var{v}$ in the \ordering\ tree, 
it uses a CAS to swing $\var{v}$'s pointer from the previous RBT root to the new RBT root.
A search in the RBT can simply read the pointer to the RBT root and perform a standard
sequential search on it.
Bashari and Woelfel ~\cite{DBLP:conf/podc/BashariW21} used persistent RBTs in a similar way for a snapshot data structure.

To prevent RBTs from growing without bound, we would like to discard
blocks that are no longer needed.
Ensuring the size of the RBT is polynomial in $p$ and $q$ will 
also keep the running time of our operations polylogarithmic.
Blocks should be kept if they contain operations still in progress.
Moreover, a block containing an \op{Enqueue}($x$) must be kept until $x$ is dequeued.

To maintain good amortized time, we periodically do a garbage collection (GC) phase.
If a \op{Refresh} on a node adds a block whose \var{index} is a multiple of $G=p^2\ceil{\log p}$, it does GC to remove obsolete blocks from the node's RBT.
To determine which blocks can be thrown away, we use a global array $\var{last}[1..p]$ where 
each process writes the index of the last block in the root
containing a null dequeue or an enqueue whose element it dequeued.
To perform GC, a process reads $\var{last}[1..p]$ and finds the maximum entry $m$.
Then, it helps complete every other process's pending dequeue 
by computing the dequeue's response and writing it in the block in the leaf that represents the dequeue.
Once this helping is complete, it follows from the FIFO property of the queue that elements enqueued 
in $\fld{root.blocks}[1..m-1]$ have all been dequeued, so GC can discard all subblocks of those.
Fortunately, there is an  RBT \op{Split} operation that can remove
these obsolete blocks from an RBT in logarithmic time \cite[Sec.~4.2]{Tar83}.

An operation $op$'s search of a RBT may fail to find the required block $B$ that has been removed 
by another process's GC phase.  If $op$
is a dequeue, $op$ must have been helped before $B$ was discarded, so $op$ can simply read its response from
its own leaf.  If $op$ is an enqueue, it can simply terminate.

Our GC scheme ensures each RBT has  $O(q+p^2\log p)$ blocks, so RBT operations take $O(\log(p+q))$ time.
Excluding GC, an operation does $O(1)$ operations on RBTs at each level of the tree for a total of
$O(\log p\log(p + q))$ steps.
A GC phase takes $O(p\log p \log(p+q))$ steps to help complete all pending operations.
If all processes carry out this GC phase, it takes a total of $O(p^2\log p\log(p+q))=O(G\log(p+q))$ steps.
Since there are at least $G$ operations between two GC phases, each node contributes $O(\log(p+q))$ steps to each operation in an amortized sense.
Adding up over all nodes an operation may have to do GC on, 
an operation spends $O(\log p\log(p+q))$ steps doing GC in an amortized sense.
So the total amortized step complexity is  $O(\log p \log(p+q))$ per operation.


\section{Future Directions}

Our focus was on optimizing step complexity for worst-case executions.
However, our queue has a higher cost than the MS-queue in the best case (when an operation
runs by itself).
Perhaps our queue could be made adaptive by having an operation capture a starting node
in the \ordering\ tree (as in \cite{DBLP:conf/stoc/AfekDT95}) rather than starting at a statically assigned leaf.
A possible application of our queue  might be to use it as the slow path in the
fast-path slow-path methodology  \cite{10.1145/2370036.2145835} to
get a queue that has good performance in practice while also having good worst-case step complexity.

A gap remains between our queue, which takes $O(\log^2 p + \log q)$ steps per operation,
and Attiya and Fouren's $\Omega(\min(c,\log\log p))$ lower bound \cite{DBLP:conf/opodis/AttiyaF17}.
It would be interesting to determine how the true step complexity of lock-free queues (or, more generally, bags)
depends on $p$.
Since a queue is also a bag, our queue is the first lock-free bag we know of that has polylogarithmic step complexity.

We believe the approach used here to implement a lock-free queue 
could be applied to obtain other lock-free
data structures with a polylogarithmic step complexity.
For example, we can easily adapt our routines to implement a  vector data structure that stores a sequence and
provides three operations: \opa{Append}{e} to add an element \var{e} to the end of the sequence,
\opa{Get}{i} to read the \var{i}th element in the sequence, and
\opa{Index}{e} to compute the position of element \var{e} in the sequence.
\here{doublecheck all preceding claims for accuracy.}
We are investigating whether a similar approach can be used for stacks, deques or even priority queues.

\begin{acks}
We thank Franck van Breugel and the anonymous referees for their helpful comments.
Funding was provided by the Natural Sciences and Engineering Research Council of Canada.
\end{acks}

\newpage

\bibliographystyle{ACM-Reference-Format}
\bibliography{queues.bib}

\newpage

\appendix

\section{Detailed Proofs for Section \ref{sec::correctness}}
\label{app::tracingDetails}

Detailed proofs omitted from Section \ref{sec::correctness} for lack of space appear here.

\sumRes*
\begin{proof}
Initially, each \fld{blocks} array  contains only an empty block $B_0$ in location 0.
By definition, $E(B_0)$ and $D(B_0)$ are empty sequences.
Moreover, $B_0.\fld{sum\sub{enq}} = B_0.\fld{sum\sub{deq}} = 0$, so the claim is true.

We show that each installation of a block $B$ into some location $\var{v}.\fld{blocks}[i]$ preserves the claim,
assuming the claim holds before this installation.  We consider two cases.

If $\var{v}$ is a leaf, $B$ was created at line \ref{enqNew} or \ref{deqNew}.
For line \ref{enqNew}, $B$ represents a single enqueue, so $|E(B)|=1$ and $|D(B)|=0$.
Since $B.\fld{sum\sub{enq}}$ is set to $\var{v}.\fld{blocks}[i-1].\fld{sum\sub{enq}}+1$ and
$B.\fld{sum\sub{deq}}$ is set to $\var{v}.\fld{blocks}[i-1].\fld{sum\sub{deq}}$, the claim follows from the hypothesis.
The proof for line~\ref{deqNew}, where $B$ has a single dequeue, is similar.

Now suppose $\var{v}$ is an internal node. By the definition of subblocks in (\ref{defsubblock}) and \Cref{lem::headProgress}, the
subblocks of $\var{v}.\fld{blocks}[1..i]$ are $\var{v}.\fld{left.blocks}[1..B.\eleft]$ 
and $\var{v}.\fld{right.blocks}[1..B.\eright]$.
Thus, the enqueues in $E(\var{v}.\fld{blocks}[0])\cdots E(\var{v}.\fld{blocks}[i])$ are those in\linebreak
$E(\var{v}.\fld{left.blocks}[0]) \cdots E(\var{v}.\fld{left.blocks}[B.\eleft])$ and those in \linebreak
$E(\var{v}.\fld{left.blocks}[0]) \cdots E(\var{v}.\fld{left.blocks}[B.\eright])$.
By the hypothesis, the total number of these enqueues is \linebreak
$\var{v}.\fld{left.blocks}[B.\eleft].\fld{sum\sub{enq}} + \var{v}.\fld{right.blocks}[B.\eright].\fld{sum\sub{enq}}$, \linebreak
which is the value that line \ref{createSumEnq} stored in $B.\fld{sum\sub{enq}}$ when $B$ was created.
The proof for \fld{sum\sub{deq}} (stored on line~\ref{createSumDeq}) is similar.
\end{proof}

\superRelationRes*
\begin{proof}
We first show that $B.\fld{super}\leq s$.
Let $R_s$ be the instance of \op{Refresh}($\var{v}.\fld{parent}$) that installs $B$'s superblock 
in $\var{v}.\fld{parent.blocks}[s]$.
By the definition of subblocks (\ref{defsubblock}), $R_s$'s read $r$ of $\var{v}.head$ at line \ref{createEndLeft} or \ref{createEndRight} obtains a value greater than $b$.
By \Cref{lem::headPosition}, $B.\fld{super}$ is not $\nl$ when $r$ occurs, which means
that $B.\fld{super}$ was set (by line~\ref{setSuper1}) to a value read from $\var{v}.\fld{parent}.\head$ before $r$.
When $r$ occurs, $\var{v}.\fld{parent.blocks}[s] = \nl$, since the later \op{CAS} by $R_s$ at line
\ref{cas} succeeds.
So, by \Cref{lem::headPosition}, $\var{v}.\fld{parent}.\head \leq s$ when $r$ occurs.
Since the value stored in $B.\fld{super}$ was read from $\var{v}.\fld{parent.head}$ before $r$ and the \head\ field is non-decreasing by \Cref{nonDecreasingHead}, it follows that $B.super\leq s$.

Next, we show that $B.\fld{super}\geq s-1$.
The value stored in $B.\fld{super}$ at line \ref{setSuper1} is read from $\var{v}.\fld{parent}.\head$ at line \ref{readParentHead} and \head\ is always at least 1, so $B.\fld{super} \geq 1$.
So, if $s\leq 2$, the claim is trivial.  Assume $s>2$ for the rest of the proof.
By \Cref{lem::headProgress}, $\var{v}.\fld{parent.blocks}[s-1]$ is not $\nl$.  Let $R_{s-1}$ be the call to
$\op{Refresh}(\var{v}.\fld{parent})$ that installed the block in $\var{v}.\fld{parent.blocks}[s-1]$.
Let $r'$ be the step when $R_{s-1}$ reads $s-1$ in $\var{v}.\fld{parent}.\head$ at line \ref{readHead}.
This read $r'$ must be before $B$ is installed in $\var{v}$;
otherwise, \Cref{successfulRefresh} would imply that $B$ is a subblock of one of 
$\var{v}.\fld{parent.blocks}[1..s-1]$, contrary to the hypothesis.
Now, consider the call to \op{Advance}($\var{v}, b$) that writes $B.\fld{super}$.
It is invoked either 
at line \ref{helpAdvance} after seeing $\var{v}.\fld{blocks}[b]\neq \nl$ at line \ref{ifHeadnotNull}
or at line \ref{advance} after ensuring $\var{v}.\fld{blocks}[b]\neq \nl$ at line~\ref{cas}.
Either way, the \op{Advance} is invoked after $B$ is installed, and therefore after $r'$.
By \Cref{nonDecreasingHead}, $\var{v}.\fld{parent}.\head$ is non-decreasing, so 
the value this \op{Advance} reads in $\var{v}.\fld{parent}.\head$ and
writes in $B.\fld{super}$ is greater than or equal to the value $s-1$ that $r'$ reads in $\var{v}.\fld{parent}.\head$.
\end{proof}

\indexDequeueRes*
\begin{proof}
We prove the claim by induction on the depth of node $\var{v}$. The base case where $\var{v}$ is the root is trivial (see Line \ref{indexBaseCase}).
Assuming the claim holds for $\var{v}$'s parent, we prove it for $\var{v}$.
Let $B=\var{v}.\fld{blocks}[b]$ and $B'$ be the superblock of $B$.
\op{IndexDequeue}($\var{v}, b, i$) first computes the index $sup$ of $B'$ in $\var{v}.\fld{parent}$.
By \Cref{superRelation}, this index is either $B.super$ or $B.super+1$.
The correct index is determined by testing on line \ref{supertest} whether $B$ is a subblock of $\var{v}.\fld{parent.blocks}[B.\fld{super}]+1$.

Next, the position of the required dequeue in $D(B')$ is computed in 
lines \ref{computeISuperStart}--\ref{computeISuperEnd}. 
We first add the number of dequeues in the subblocks of $B'$ in $\var{v}$ that precede $B$ on line \ref{computeISuperStart}.
If $\var{v}$ is the right child of its parent, then all of the subblocks of $B'$ from $\var{v}$'s left sibling
also precede the required dequeue, so we add the number of dequeues in those subblocks in line \ref{considerLeftBeforeRight}.

Finally, \op{IndexDequeue} is called recursively on $\var{v}$'s parent.
Since $B$ has been propagated to the root, so has its superblock $B'$.
Thus, all preconditions of the recursive call are met.
By the induction hypothesis, the recursive call returns the location of the required dequeue in the root.\end{proof}

\getEnqRes*
\begin{proof}
We prove the claim by induction on the height of node $\var{v}$.
If $\var{v}$ is a leaf, the hypothesis implies that $i=1$ and the block $\var{v}.\fld{blocks}[b]$ represents 
an enqueue whose argument is stored in $\var{v}.\fld{blocks}[b].\fld{element}$.
\op{GetEnqueue} returns the argument of this enqueue at line \ref{getBaseCase}.

Assuming the claim holds for $\var{v}$'s children, we prove it for $\var{v}$.
Let $B$ be $\var{v}.\fld{blocks}[b]$.
By (\ref{defSeqs}),
$E(B)$ is obtained by concatenating the enqueue sequences of the direct subblocks
of $B$, which are listed in (\ref{defsubblock}).
By \Cref{lem::sum}, $\var{sum\sub{left}}-\var{prev\sub{left}}$ is the number
of enqueues in $E(B)$ that come from $B$'s subblocks in $\var{v}$'s left child.
Thus, $dir$ is set to the direction for the child of $\var{v}$ that contains the required enqueue.
Moreover, when line \ref{endChooseDir} is reached, $i$ is the position of the required enqueue within the portion $E'$ of $E(B)$ that comes from that child.
Thus,  line \ref{getChild} finds the index $b'$ of the subblock $B'$ containing the required enqueue.
By \Cref{lem::sum}, $\var{v}.\var{dir.blocks}[b'-1].\var{sum\sub{enq}} - \var{prev\sub{dir}}$ is the number of 
enqueues in $E'$ before the enqueues of block $B'$, so
the value $i'$ computed on line \ref{getChildIndex} is the position of the required enqueue within $E(B')$.
Thus, the recursive call on line \ref{getRecurse} satisfies its precondition, and 
returns the required result, by the induction hypothesis.
\end{proof}

\sizeCorrectRes*
\begin{proof}
We prove the claim by induction on $b$. 
The base case when ${b=0}$ is trivially true, since the queue is initially empty and 
$\var{root}.\fld{blocks}[0]$ contains an empty block whose \size\ field is $0$. 
Assuming the claim holds for $b-1$, we prove it for $b$.
The \fld{size} field of the block $B$ installed in $\var{root}.\fld{blocks}[b]$ is computed
at line \ref{computeLength} of a call to \op{CreateBlock}(\var{root, b}).
By the induction hypothesis, $root.\fld{blocks}[b-1].\size$ gives the size of the queue before the operations
of block $B$ are performed.
By \Cref{lem::sum}, the values of \var{num\sub{enq}} and \var{num\sub{deq}}
are the number of enqueues and dequeues contained in $B$.
Hence, the size of the queue after the operations of $B$ are performed (with enqueues before dequeues as specified by $L$)
is $\max(0, \var{root}.\fld{blocks}[b-1].\size + \var{num\sub{enq}} - \var{num\sub{deq}})$.
\end{proof}


\section{Detailed Description of the Bounded-Space Implementation}
\label{reducing-details}

\renewcommand{\algorithmiccomment}[1]{\hfill\eqparbox{COMMENTSINGLEAPP}{\com\ #1}}

Here, we give more details of the bounded-space construction sketched in Section \ref{reducing}.
We first describe the modifications to the implementation.
To avoid confusion, we use nodes to refer to the nodes of the ordering tree, and blocks to refer
to the nodes of a RBT (since the RBT stores blocks).

The space-bounded implementation uses two additional shared arrays:
the \var{leaf} array allows processes to access one another's leaves to perform helping, and
the \var{last} array is used to determine which blocks are safe to discard.
\begin{itemize}
\item \typ{Node[]} \var{leaf}$[1..p]$ \Comment{\var{leaf}$[k]$ is the leaf assigned to process $k$}
\item \typ{int[]} \var{last}$[1..p]$ \Comment{\var{last}$[k]$ is the largest \fld{index} of a block in the}\\
\mbox{ }\Comment{root that process $k$ has observed to contain either an en-}\\
\mbox{ }\Comment{queue whose value has been dequeued or a null dequeue }
\end{itemize}

The \fld{blocks} field of each node in the \ordering\ tree is implemented as a pointer to the root of a RBT of \typ{Blocks} rather than an infinite array.  
Each RBT is initialized with an empty block with index~0.
Any access to an entry of the \fld{blocks} array is replaced by a search in the RBT.
The node's \head\ field, which previously gave the next position to insert into the \fld{blocks} array is no
longer needed; we can instead simply find the maximum \fld{index} of any block in the RBT.
To facilitate this, \op{MaxBlock} is a query operation on the RBT that 
returns the block with the maximum \fld{index}.
We can store, in the root of the RBT, a pointer to the maximum block so that \op{MaxBlock}
can be done in constant time, without affecting the time of other RBT operations.
Similarly, a \op{MinBlock} query finds the block with the minimum \fld{index} in a RBT.

Blocks no longer require the \fld{super} field.  It was used to quickly find a block's 
superblock in the parent node's \fld{blocks} array, but this can now be done efficiently 
by searching the parent's \fld{blocks} RBT instead.
Each \typ{Block} has an additional field.
\begin{itemize}
\item \typ{int} \var{index} \Comment{position this block would have in the \fld{blocks} array}
\end{itemize}
To facilitate helping, each \typ{Block} in a leaf has one more additional field that is used only for blocks that store a dequeue operation.
\begin{itemize}
\item \typ{Object} \var{response} \Comment{response of the dequeue in the block}
\end{itemize}

\here{Talk about what happens if a block is not found in an RBT}

Pseudocode for the space-bounded implementation appears in Figures \ref{pseudocode1GC} and \ref{pseudocode2GC}.
New or modified code appears in blue.
The \op{Propagate} and \op{GetEnqueue} routines are unchanged.
A few lines have been added to \op{FindResponse} to update the \var{last} array
to ensure that it stores the value described above.
Minor modifications have also
been made to \op{Enqueue}, \op{Dequeue}, \op{CreateBlock}, \op{Refresh} and \op{Append}
to accommodate  the switch from an array of blocks to a RBT of blocks (and the corresponding disappearance
of the \head\ field).
In addition, the second half of the \op{Dequeue} routine is now in a
separate routine called \op{CompleteDeq} so that it can also be used by other processes
helping to complete the operation.
The \op{Refresh} routine no longer needs to set the \fld{super} field of blocks since
that field has been removed.
The \op{IndexDequeue} routine, which must trace the location of a dequeue along
a path from its leaf to the root has a minor modification to search the \fld{blocks} RBT at 
each level instead of using the \fld{super} field.

\renewcommand{\algorithmiccomment}[1]{\hfill\eqparbox{COMMENTDOUBLE}{\com\ #1}}

\begin{figure*}
\begin{minipage}[t]{0.41\textwidth}
\begin{algorithmic}[1]
\setcounter{ALG@line}{200}

\Function{void}{Enqueue}{\typ{Object} \var{e}} 
    \gc{
    \State \var{h} \assign \Call{MaxBlock}{\var{leaf}.\fld{blocks}}.\fld{index}+1
	\State \hangbox{let \var{B} be a new \typ{Block} with \fld{element} \assign\ \var{e},\\
		$\fld{sum\sub{enq}} \assign\ \var{leaf}.\fld{blocks}[h-1].\fld{sum\sub{enq}}+1,$\\
		$\fld{sum\sub{deq}} \assign\ \var{leaf}.\fld{blocks}[h-1].\fld{sum\sub{deq}}$,\\
		$\fld{index}\assign h$}\label{enqNewGC}
	}
    \State \Call{Append}{\var{B}}
\EndFunction{Enqueue}

\spac

\Function{Object}{Dequeue()}{} 
    \gc{
    \State \var{h} \assign \Call{MaxBlock}{\var{leaf}.\fld{blocks}}.\fld{index}+1
    \State \hangbox{let \var{B} be a new \typ{Block} with \fld{element} \assign\ \nl,\\
	    $\fld{sum\sub{enq}} \assign\ \var{leaf}.\fld{blocks}[h-1].\fld{sum\sub{enq}},$\\
	    $\fld{sum\sub{deq}} \assign\ \var{leaf}.\fld{blocks}[h-1].\fld{sum\sub{deq}}+1$,\\
	    $\fld{index}\assign h$}\label{deqNewGC}
	}
    \State \Call{Append}{\var{B}}
    \gc{\State \Return \Call{CompleteDeq}{\var{leaf}, $h$}}
\EndFunction{Dequeue}

\spac

\gc{
\Function{Object}{CompleteDeq}{\var{leaf}, $h$} 
    \State \linecomment finish propagated dequeue in $\var{leaf}.\fld{blocks}[h]$
    \State $\langle \var{b}, \var{i}\rangle \assign\ \Call{IndexDequeue}{\var{leaf}, h, 1}$
    \State $\var{res} \assign \Call{FindResponse}{b, i}$
	\State \Return{\var{res}}
\EndFunction{CompleteDeq}
}

\spac

\Function{void}{Append}{\typ{Block} \var{B}}  \linecomment{append block to leaf and propagate to root}
    \gc{\State \var{leaf}.\fld{blocks} \assign \Call{AddBlock}{\var{leaf}, \var{leaf}.\fld{blocks}, \var{B}}\label{appendLeafGC}}
    \State \Call{Propagate}{\var{leaf}.\fld{parent}} 
\EndFunction{Append}

\spac

\gc{
\Function{RBT}{AddBlock}{\typ{Node} \var{v}, \typ{RBT} \var{T}, \typ{Block} \var{B}} 
    \State \linecomment{add block $\var{B}\neq\nl$ to \var{T}; do GC if necessary}
    \If{\var{B}.\fld{index} is a multiple of $G$}
    	\State \linecomment{Do garbage collection}\label{GCstart}
        \State $s \assign \Call{SplitBlock}{\var{v}}.\var{index}$\label{callSplitBlock}
        \State \Call{Help}{}\label{invokeHelp}
        \State $T' \assign\ \Call{Split}{T,s}$\label{splitGC}
        \State \linecomment{\op{Split} removes blocks with $\fld{index} < s$}
        \State \Return \Call{Insert}{$T',B$}\label{GCend}
	\Else\ \Return \Call{Insert}{$T, B$}
    \EndIf
\EndFunction{AddBlock}
}

\spac

\gc{
\Function{Block}{SplitBlock}{\typ{Node} v}
	\State \linecomment{figure out where to split \var{v}'s RBT}
	\If{$\var{v}=\var{root}$}
		\State $m \assign 0$
		\For{$k\assign 1..p$} $m\assign \max(m, \var{v}.\fld{last}[k])$
		\EndFor
		\State $B\assign \var{root}.\fld{blocks}[m-1]$\label{SBlookup1}
	\Else
		\State $B_p \assign \Call{SplitBlock}{\var{v}.\fld{parent}}$\label{SBrecurse}
		\State \var{dir} \assign ($\var{v}=\var{v}.\fld{parent.left}$ ? \fld{left} : \fld{right})
		\State $B \assign \var{v}.\fld{blocks}[B_p.\edir]$\label{SBlookup2}
	\EndIf
	\State \linecomment If $B$ was discarded, use leftmost block instead
	\State \Return ($B = \nl$ ? $\Call{MinBlock}{\var{v}.\var{blocks}}$ : $B$) \label{substitute}
\EndFunction{SplitBlock}
}

\spac

\Function{void}{Propagate}{\typ{Node} \var{v}}
    \State \linecomment{propagate blocks from \var{v}'s children to root}
    \If{\bf{not} \Call{Refresh}{\var{v}}} \label{firstRefreshGC}  \hfill \com\ double refresh
        \State \Call{Refresh}{\var{v}} \label{secondRefreshGC}
    \EndIf
    \If{$\var{v} \neq \var{root}$} \hfill \com\ recurse up tree
        \State \Call{Propagate}{\var{v}.\fld{parent}}
    \EndIf
\EndFunction{Propagate}

\spac

\Function{boolean}{Refresh}{\typ{Node} \var{v}} \linecomment{try to append a new block to $\var{v}.\fld{blocks}$}
	\gc{
	\State $T \assign \var{v}.\fld{blocks}$\label{refreshReadTGC}
    \State $\var{h} \assign \Call{MaxBlock}{T}.\fld{index}+1$ \label{refreshReadMax}
    }
    \State \var{new} \assign\ \Call{CreateBlock}{\var{v, h}} \label{invokeCreateBlockGC}
    \If{\var{new = \nl}} \Return{\tr} \label{addOPGC}
    \Else  
    	\gc{
	    \State $T' \assign \Call{AddBlock}{v, T, \var{new}}$\label{refreshABGC}
		\State \Return \Call{CAS}{\var{v}.\fld{blocks}, $T$, $T'$}\label{casGC}
		}
    \EndIf

\EndFunction{Refresh}

\end{algorithmic}
\end{minipage}
\hfill
\begin{minipage}[t]{0.56\textwidth}

\begin{algorithmic}[1]
\setcounter{ALG@line}{267}

\gc{
\Function{boolean}{Propagated}{\typ{Node} \var{v}, \typ{int} \var{b}}  \linecomment{check if $\var{v}.\fld{blocks}[b]$ has  propagated to \var{root}}
    \State \linecomment{Precondition:  $\var{v}.\fld{blocks}[b]$ exists}
    \If {$\var{v} = \var{root}$} \Return{true}
    \Else
    	\State $T\assign \var{v}.\fld{parent.blocks}$
		\State \fld{dir} \assign\ (\var{v}.\fld{parent.left} = \var{v} ? \fld{left} : \fld{right}) 
		\If{$\Call{MaxBlock}{T}.\edir < b$} \Return \fa
		\Else
			\State $B_p$ \assign min block in $T$ with $\edir \geq b$\label{searchsuper3}
			\State \Return \Call{Propagated}{\var{v}.\fld{parent}, $B_p$.\fld{index}}
		\EndIf
	\EndIf
\EndFunction{Propagated}
}

\spac

\Function{$\langle\typ{int}, \typ{int}\rangle$}{IndexDequeue}{\typ{Node} \var{v}, \typ{int} \var{b}, \typ{int} \var{i}}
    \State \linecomment{return $\langle\var{x, y}\rangle$ such that \var{i}th dequeue in $D(\var{v}.\fld{blocks}[\var{b}])$ is \var{y}th dequeue of $D(\var{root}.\fld{blocks}[\var{x}])$}
    \State \linecomment{Precondition: \var{v}.\fld{blocks}[\var{b}] exists and has propagated to root and $|D(\var{v}.\fld{blocks}[\var{b}])|\geq \var{i}$}
    \If{$\var{v} = \var{root}$} \Return $\langle\var{b, i}\rangle$ \label{indexBaseCaseGC}
    \Else
	    \State \fld{dir} \assign\ (\var{v}.\fld{parent.left} = \var{v} ? \fld{left} : \fld{right}) 
    	\gc{
		\State $T\assign \var{v}.\fld{parent.blocks}$
		\State $B_p$ \assign min block in $T$ with $\edir \geq b$\label{searchsuper1}
		\State $B_p'$ \assign max block in $T$ with $\edir < b$\label{searchsuper2}
	    }
	    \State \linecomment{compute index \var{i} of dequeue in superblock $B_p$}
	    \State \hangbox{\var{i} += $\var{v}.\fld{blocks}[\var{b}-1].\fld{sum\sub{deq}} -
	    		\var{v}.\fld{blocks}[\gc{B_p'}.\edir].\fld{sum\sub{deq}}$}
        \If{$\fld{dir} = \fld{right}$} 
        	\State \hangbox{\var{i} += $\var{v}.\fld{blocks}[\gc{B_p}.\eleft].\fld{sum\sub{deq}} - \mbox{ }
					\var{v}.\fld{blocks}[\gc{B_p'}.\eleft].\fld{sum\sub{deq}}$}\label{considerLeftBeforeRightGC}
        \EndIf \label{computeISuperEndGC}
        \State \Return\Call{IndexDequeue}{\var{v}.\fld{parent}, $\gc{B_p.\fld{index}}$, \var{i}}
    \EndIf
\EndFunction{IndexDequeue}

\spac

\gc{
\Function{void}{Help}{} \linecomment{help pending operations}
    
    \For{$\ell$ in $\var{leaf}[1..k]$}
        \State $\var{B} \assign \Call{MaxBlock}{\ell.\fld{blocks}}$
        \If{$\var{B}.\fld{element} = \nl$ and $\var{B}.\fld{index}>0$ and \Call{Propagated}{$\ell, \var{B}.\fld{index}$}} 
            \State \linecomment{operation is a propagated dequeue}
            \State $\var{B}.\fld{response} \assign\ \Call{CompleteDeq}{\ell, \var{B}.\fld{index}}$ 
        \EndIf
    \EndFor
\EndFunction{Help}
}
\spac

\Function{Block}{CreateBlock}{\typ{Node} \var{v}, \typ{int} \var{i}} \linecomment{create new block  to install in \var{v}.\fld{blocks}[\var{i}]}
    \State let \var{new} be a new \typ{Block} \label{initNewBlockGC}
    \gc{
    \State \var{new}.\eleft \assign\ $\Call{MaxBlock}{\var{v}.\fld{left}.\fld{blocks}}.\fld{index}$\label{createEndLeftGC}
    \State \var{new}.\eright \assign\ $\Call{MaxBlock}{\var{v}.\fld{right}.\fld{blocks}}.\fld{index}$\label{createEndRightGC}
    \State $\var{new}.\fld{index} \assign i$\label{createIndexGC}
    }
	\State \var{new}.\fld{sum\sub{enq}} \assign\ \var{v}.\fld{left.blocks}[\var{new}.\eleft].\fld{sum\sub{enq}} + 
			\var{v}.\fld{right.blocks}[\var{new}.\eright].\fld{sum\sub{enq}}
	\State \var{new}.\fld{sum\sub{deq}} \assign\ \var{v}.\fld{left.blocks}[\var{new}.\eleft].\fld{sum\sub{deq}} +
			\var{v}.\fld{right.blocks}[\var{new}.\eright].\fld{sum\sub{deq}}
    \State \var{num\sub{enq}} \assign\ $\var{new}.\fld{sum\sub{enq}} - \var{v}.\fld{blocks}[\var{i}-1].\fld{sum\sub{enq}}$\label{computeNumEnqGC}
    \State \var{num\sub{deq}} \assign\ $\var{new}.\fld{sum\sub{deq}} - \var{v}.\fld{blocks}[\var{i}-1].\fld{sum\sub{deq}}$
    \If{$\var{v} = \var{root}$}
        \State \var{new}.\fld{size} \assign\ max(0, $\var{v}.\fld{blocks}[\var{i}-1].\size\ + 
        	\var{num\sub{enq}} - \var{num\sub{deq}}$)\label{computeLengthGC}
    \EndIf
    \If{$\var{num\sub{enq}} + \var{num\sub{deq}} = 0$}\label{testEmptyGC}
        \State \Return \nl \hfill \com\ no blocks to be propagate to \var{v}\label{CBnullGC}
    \Else
        \State \Return \var{new}
    \EndIf
\EndFunction{CreateBlock}

\spac 

\Function{element}{FindResponse}{\typ{int} \var{b}, \typ{int} \var{i}} \linecomment{find response to \var{i}th dequeue in $D(\var{root}.\fld{blocks}[\var{b}])$}
    \State \linecomment{Precondition:  $1\leq i\leq |D(\var{root}.\fld{blocks}[\var{b}])|$}
    \State \var{num\sub{enq}} \assign\ $\var{root}.\fld{blocks}[\var{b}].\fld{sum\sub{enq}} - \mbox{ }$
    		$\var{root}.\fld{blocks}[\var{b}-1].\fld{sum\sub{enq}}$\label{FRnumenqGC}
    \If{$\var{root}.\fld{blocks}[\var{b}-1].\size + \var{num\sub{enq}} < \var{i}$}\label{checkEmptyGC} \linecomment{queue is empty when dequeue occurs}
        \gc{
        \If{$b > \var{last}[\var{id}]$} $\var{last}[\var{id}]\assign b$
        \EndIf
        }
        \State \Return \nl \hfill \label{returnNullGC}
    \Else \ \linecomment{response is the \var{e}th enqueue in the root}
        \State \var{e} \assign\ \var{i} + \var{root}.\fld{blocks}[\var{b}-1].\fld{sum\sub{enq}} - 
			\var{root}.\fld{blocks}[\var{b}-1].\size\label{computeEGC}
		 
		\State find min $b_e \leq \var{b}$ with $\var{root}.\fld{blocks}[b_e].\fld{sum\sub{enq}} \geq \var{e}$\label{FRsearchGC} \linecomment{use binary search}
		\State $i_e \assign\ \var{e} - \var{root}.\fld{blocks}[b_e-1].\fld{sum\sub{enq}}$ \linecomment{find rank of enqueue within its block}
		\gc{
        \State \var{res} \assign \Call{GetEnqueue}{\var{root}, $b_e$, $i_e$}\label{findAnswerGC}
        \If {$b_e > \var{last}[\var{id}]$} $\var{last}[\var{id}]\assign b_e$
        \EndIf
        \State \Return \var{res}
		}
    \EndIf
\EndFunction{FindResponse}

\end{algorithmic}
\end{minipage}
\vspace*{-4mm}
\caption{Bounded-space queue implementation.\label{pseudocode1GC}  $G$ is a constant, which we choose to be $p^2 \ceil{\log p}$.}
\end{figure*}

\begin{figure*}
\begin{algorithmic}[1]
\setcounter{ALG@line}{341}

\Function{element}{GetEnqueue}{\typ{Node} \var{v}, \typ{int} \var{b}, \typ{int} \var{i}} \Comment{returns argument of \var{i}th enqueue in $E(\var{v}.\fld{blocks}[\var{b}])$}
    \State \linecomment{Preconditions: $\var{i}\geq 1$ and \var{v}.\fld{blocks}[\var{b}] exists and contains at least \var{i} enqueues}
    \If{\var{v} is a leaf node} \Return \var{v}.\fld{blocks}[\var{b}].\fld{element} \label{getBaseCaseGC}
    \Else 
        \State \var{sum\sub{left}} \assign\ \var{v}.\fld{left.blocks}[\var{v}.\fld{blocks}[\var{b}].\eleft].\fld{sum\sub{enq}} \Comment{\#\ of enqueues in \var{v}.\fld{blocks}[1..$\var{b}$] from left child}
        \State \var{prev\sub{left}} \assign\ \var{v}.\fld{left.blocks}[\var{v}.\fld{blocks}[$\var{b}-1$].\eleft].\fld{sum\sub{enq}} \Comment{\#\ of enqueues in \var{v}.\fld{blocks}[1..$\var{b}-1$] from left child}
        \State \var{prev\sub{right}} \assign\ \var{v}.\fld{right.blocks}[\var{v}.\fld{blocks}[$\var{b}-1$].\eright].\fld{sum\sub{enq}} \Comment{\#\ of enqueues in \var{v}.\fld{blocks}[1..$\var{b}-1$] from right child}
        \If{$\var{i} \leq \var{sum\sub{left}} - \var{prev\sub{left}}$} \label{leftOrRightGC} \cmt{required enqueue is in \var{v}.\fld{left}}
            \State \fld{dir} \assign\ \fld{left}
        \Else \Comment{required enqueue is in \var{v}.\fld{right}}
            \State \fld{dir} \assign\ \fld{right}
            \State $\var{i}\ \assign\ \var{i} - (\var{sum\sub{left}} - \var{prev\sub{left}})$
        \EndIf
        \State \linecomment{Use binary search to find enqueue's block in \var{v}.\fld{dir.blocks} and its rank within block}
        \State find minimum $\var{b}'$ in range [\var{v}.\fld{blocks}[\var{b}-1].\edir+1..\var{v}.\fld{blocks}[\var{b}].\edir] s.t. $\var{v}.\fld{dir.blocks}[\var{b}'].\fld{sum\sub{enq}} \geq \var{i} + \var{prev\sub{dir}}$\label{getChildGC}
        \State $\var{i}'$ \assign\ $\var{i} - (\var{v}.\fld{dir.blocks}[\var{b}'-1].\fld{sum\sub{enq}} - \var{prev\sub{dir}})$
        \State \Return\Call{GetEnqueue}{\var{v}.\fld{dir}, $\var{b}'$, $\var{i}'$}
    \EndIf
\EndFunction{GetEnqueue}

\end{algorithmic}

\caption{Bounded-space queue implementation, continued.\label{pseudocode2GC}}
\end{figure*}

The new routines, \op{AddBlock}, \op{SplitBlock}, \op{Help} and \op{Propagated} are
used to implement the garbage collection (GC) phase.
When a \op{Refresh} or \op{Append} wants to add a new block to a node's \fld{blocks} RBT,
it calls the new \opemph{AddBlock} routine.
Before attempting to add the block to a node's RBT, \op{AddBlock} triggers a GC phase on the RBT if 
the new block's \fld{index} is a multiple of the constant $G$, which we choose to be $p^2\ceil{\log p}$.
This ensures that obsolete blocks are removed from the RBT once every $G$ times a new block is added to it.
The GC phase uses \op{SplitBlock} to determine the index $s$ of the oldest block to keep,
calls \op{Help} to help all pending dequeues that have been propagated to the root
(to ensure that all blocks before $s$ can safely be discarded),
uses the standard RBT \op{Split} routine \cite{Tar83}
to remove all blocks with \fld{index} less than $s$,
inserts the new block
and finally performs a \op{CAS} to swing the node's \fld{blocks} pointer to the new RBT.

To determine the oldest block in a node \var{v} to keep, the \opemph{SplitBlock} routine  first
uses the \var{last} array to
find the most recent block $B_{root}$ in the root that contains either an enqueue that has been dequeued
or a null dequeue.  By the FIFO property of queues, all enqueues in blocks before $B_{root}$ are either
dequeued or will be dequeued by a dequeue that is currently in progress.  Once those pending dequeues have been
helped to complete by line \ref{invokeHelp},
it is safe to discard any blocks in the root older than $B_{root}$, as well
as their subblocks.\footnote{If we used the more conservative approach of discarding blocks whose indices 
are smaller than the \emph{minimum} entry of \var{last} instead of the maximum, helping would be unnecessary, but then one slow process could prevent GC from discarding any blocks, so the space would not be bounded.}
The \op{SplitBlock} uses the \eleft\ and \eright\ fields to find the last block in \var{v} that is a subblock
of $B_{root}$ (or any older block in the root, in case $B_{root}$ has no subblocks in \var{v}).
As \op{SplitBlock} is in progress, it is possible that some block that it needs in a node $\var{v}'$
along the path from \var{v} to the root is discarded by another GC phase.
In this case, \op{SplitBlock} uses the last subblock in \var{v} of the oldest block in $\var{v}'$ instead
(since a GC phase on $\var{v}'$ determined that all blocks older than that are safe to discard anyway).

The \opemph{Help} routine is fairly straightforward:  it loops through all leaves and
helps the dequeue that is in progress there if it has already been propagated to the root.
The \opemph{Propagated} function is used to determine whether the dequeue has propagated to the root.

In the code, we use $\var{v}.\fld{blocks}[i]$ to refer to the block in the RBT stored in $\var{v}.\fld{blocks}$
with index~$i$.
A search for this block may sometimes not find it, if it has already been discarded
by another process's GC phase.
As mentioned in Section \ref{reducing}, if this happens to an enqueue operation,
the enqueue can simply terminate because the fact that the block is gone means that another
process has helped the enqueue reach the root of the ordering tree.
Similarly, if  a dequeue operation performs a failed search on a RBT, the dequeue can return the value 
written in the \fld{response} field of the 
leaf block that represents the dequeue and terminate, since some other process
will have written the \fld{response} there before discarding the needed block.
We do not explicitly write this early termination in the pseudocode every time we do a lookup in an RBT.
There is one exception to this rule:  if an RBT lookup for block $B$ returns \nl\
on line \ref{SBlookup1} or \ref{SBlookup2} of
\op{SplitBlock} because the required block has been discarded, 
we continue doing GC, since we do not want a GC phase on one node
to be prevented from cleaning up its RBT because a GC phase on a different node threw away some blocks
that were needed.
Line \ref{substitute} says what to do in this case.


\subsection{Correctness}
There are enough changes to the algorithm that a new proof of correctness
is required.  Its structure mirrors the proof of the original algorithm, but requires additional
reasoning to ensure GC does not interfere with other routines.

\subsubsection{Basic Properties}

The following observation describes
how the set of blocks in a node's RBT can be modified.

\begin{lemma}\label{RBTupdates}
Suppose a step of the algorithm changes $\var{v}.\fld{blocks}$ from a non-empty tree $T$ to $T'$.
If the set of \fld{index} values  in $T$ is $I$, then the set of \fld{index} values in $T'$ is 
$(I \cap [m-1,\infty))\cup \{\max(I) + 1\}$ for some~$m$.
\end{lemma}
\begin{proof}
The RBT of a node is updated only at line \ref{appendLeafGC} or \ref{casGC}.

If line \ref{appendLeafGC} of an \op{Append} operation modifies $\var{v}.\fld{blocks}$,
then \var{v} is a leaf node, and no other process ever modifies $\var{v}.\fld{blocks}$.
$T'$ was obtained from $T$ by calling \op{AddBlock}($T,B$).
$B$ was created either by the \op{Enqueue} or \op{Dequeue} that called \op{Append}.
Either way, $B.\fld{index} = \max(I)+1$.
The \op{AddBlock} that creates $T'$ may optionally \op{Split} the RBT and then add $B$ to it.
So the claim is satisfied.

If line \ref{casGC} of a \op{Refresh} modifies $\var{v}.\fld{blocks}$, then \var{v} is an internal node.
After reading $T$ from $\var{v}.\fld{blocks}$ at line \ref{refreshReadTGC},
the \op{Refresh} then creates the  block \var{new},
and then calls \op{AddBlock}($T,new$) to create $T'$.
Line \ref{createIndexGC} sets $\var{new}.\var{index} = \max(I)+1$.
The \op{AddBlock} that creates $T'$ may optionally \op{Split} the RBT and then add $new$ to it.
So the claim is satisfied.
\end{proof}

Since each RBT starts with a single block with \fld{index} 0, the following is an easy consequence of \Cref{RBTupdates}.
\begin{corollary}\label{nonEmptyIncreasing}
The RBT stored in each node \var{v} is never empty and always stores a set of blocks with consecutive indices.
Moreover, its maximum \fld{index} can only increase over time.
\end{corollary}

Since RBTs are always non-empty, calls to \op{MaxBlock} have well-defined answers.  
Throughout the proof, we use $\var{v}.\blocks[b]$ to refer to the block with \fld{index} $b$ that appeared
in \var{v}'s tree at some time during the execution.
It follows from \Cref{RBTupdates} and \Cref{nonEmptyIncreasing} that 
each time a new block appears in \var{v}'s RBT, its \fld{index} is greater than any block
that has appeared in \var{v}'s RBT earlier.
Thus, $\var{v}.\blocks[b]$ is unique, if it exists.
We also use this notation in the code to indicate that a search of the RBT should be performed for the block
with \fld{index} $b$.

We now establish that Definition (\ref{defsubblock}) of a block's subblocks
still makes sense by proving the analogue of Lemma \ref{lem::headProgress}.


\begin{customlemma}{\ref{lem::headProgress}$'$}\label{lem::headProgressGC}
 If $h>0$ and a block with index $h$ has been inserted into $\var{v}.\fld{blocks}$ then 
 $\var{v}.\fld{blocks}[h-1].\fld{end\sub{left}} \leq \var{v}.\fld{blocks}[h].\fld{end\sub{left}}$ and 
 $\var{v}.\fld{blocks}[h-1].\fld{end\sub{right}} \leq \var{v}.\fld{blocks}[h].\fld{end\sub{right}}$.
\end{customlemma}
\begin{proof}
The block $B$ with index $h$ was installed into \var{v}'s RBT by the \op{CAS} at line \ref{casGC}.
Suppose that \op{CAS} changed the tree from $T$ to $T'$.
Before this \op{CAS}, line \ref{refreshReadTGC} read the tree $T$ from $\var{v}.blocks$,
line \ref{refreshReadMax} found a block $B'$  with \fld{index} $h-1$ in $T$,
and then line \ref{invokeCreateBlockGC} created the block $B$ with $\fld{index}=h$.
Since $B'$ was already in $T$ before $B$ was created, the
\op{CreateBlock}$(\var{v},b-1)$ that created $B'$ terminated before the 
\op{CreateBlock}$(\var{v},b)$ that created $B$ started.
By \Cref{nonEmptyIncreasing}, the value that line \ref{createEndLeftGC} of \op{CreateBlock}$(\var{v},b-1)$ 
stores in $B'.\eleft$ is less than or equal to the value that line \ref{createEndLeftGC} 
of \op{CreateBlock}$(\var{v},b)$ stores in $B.\eleft$.  
Similarly, the values stored in $B'.\eright$ and $B.\eright$ at line \ref{createEndRightGC} satisfy the claim.
\end{proof}

\Cref{lem::headProgressGC} implies that 
the nodes of an in-order traversal of any RBT have non-decreasing values of
$\eleft$ (and of $\eright$).
Thus, the searches for a block based on \eleft\ or \eright\ values 
at lines \ref{searchsuper3}, \ref{searchsuper1} and \ref{searchsuper2},
which are used to look for the superblock of a node or its predecessor, can be done in logarithmic time.

\Cref{lem::subblocksDistinct}, \Cref{lem::noDuplicates}, \Cref{lem::sum}, \Cref{blockNotEmpty} 
hold for the modified algorithm.  Their proofs are identical to those given in 
Section \ref{sec::basicProperties} since they  depend only on \Cref{lem::headProgress} (which can be replaced by \Cref{lem::headProgressGC})
and the definition of subblocks given in (\ref{defsubblock}).
In particular, \Cref{lem::sum} says that nodes in an in-order traversal of a RBT have non-decreasing
values of \fld{sum\sub{enq}}  so the searches of a RBT based on
\fld{sum\sub{enq}} values in lines
\ref{FRsearchGC} and \ref{getChildGC} can be done in logarithmic time.

\subsubsection{Propagating Operations to the Root}

Next, we prove an analogue of \Cref{successfulRefresh}.
We say a node \var{v} \emph{contains} an operation
if some block containing the operation has previously appeared in \var{v}'s RBT (even if 
the block has been removed from the RBT by a subsequent \op{Split}).

\begin{customlemma}{\ref{successfulRefresh}$'$}
\label{successfulRefreshGC}
Let $R$ be a call to \op{Refresh}(\var{v}) that performs a successful \op{CAS} on line \ref{casGC} (or terminates at line \ref{addOPGC}).
In the configuration after that CAS (or termination, respectively), \var{v} contains all operations that \var{v}'s children contained 
when $R$ executed line~\ref{refreshReadTGC}.
\end{customlemma}
\begin{proof}
Suppose \var{v}'s child (without loss of generality, $\var{v}.\fld{left}$) contained an operation $op$ 
when $R$ executed line \ref{refreshReadTGC}.
Let $i$ be the index of the block containing $op$ that was in \var{v}'s RBT before 
$R$ executed line \ref{refreshReadTGC}.
We consider two cases.

Suppose $R$'s call to \op{CreateBlock} returns a new block $B'$
that is installed in $\var{v}.\fld{blocks}$
by $R$'s \op{CAS} at line \ref{casGC}.
The \op{CreateBlock} set $B'.\eleft$ to the maximum \fld{index} in $\var{v}.\fld{left}$'s RBT at line \ref{createEndLeftGC}.
By \Cref{nonEmptyIncreasing}, this maximum \fld{index} is bigger than $i$.
By the definition of subblocks, some block in \var{v} contains $B$ as a subblock
and therefore \var{v} contains $op$.

Now suppose $R$'s call to \op{CreateBlock} returns \nl, causing $R$ to terminate at line \ref{addOPGC}.
Let $h$ be the maximum \fld{index} in $T$ plus 1.
By reasoning identical to the last paragraph of \Cref{successfulRefresh}'s proof,
it follows from the fact that $\var{num\sub{enq}}+\var{num\sub{deq}}=0$ at line \ref{testEmptyGC}
that the blocks $\var{v}.\fld{left.blocks}[\var{v}.\fld{blocks}[h-1].\eleft+1..new.\eleft]$ and
$\var{v}.\fld{right.blocks}[\var{v}.\fld{blocks}[h-1].\eright+1..new.\eright]$
contain no operations.
By \Cref{blockNotEmpty}, each block contains at least one operation, so
these ranges must be empty, and $\var{v}.\fld{blocks}[h-1].\eleft \geq \var{new}.\eleft\geq i$.
This implies that the block $B$ containing $op$ is a subblock of some block that has appeared
in \var{v}'s RBT, so $op$ is contained in \var{v}.
\end{proof}

This allows us to show that a double \op{Refresh} propagates operations up the tree, 
as in \Cref{lem::doubleRefresh}.

\begin{customlemma}{\ref{lem::doubleRefresh}$'$}\label{lem::doubleRefreshGC}
Consider two consecutive terminating calls $R_1$, $R_2$ to \op{Refresh}(\var{v}) by the same process.
All operations contained \var{v}'s children when $R_1$ begins
are contained in \var{v} when $R_2$ terminates.
\end{customlemma}
\begin{proof}
If either $R_1$ or $R_2$ performs a successful \op{CAS} at line \ref{casGC} or terminates at line \ref{addOPGC}, the claim follows
from \Cref{successfulRefreshGC}.
So suppose both $R_1$ and $R_2$ perform a failed \op{CAS} at line \ref{casGC}.
Then some other \op{CAS} on $\var{v}.\fld{blocks}$ succeeds between the time each \op{Refresh}
reads $\var{v}.\fld{blocks}$ at line \ref{refreshReadTGC} and performs its \op{CAS} at line \ref{casGC}.
Consider the \op{Refresh} $R_3$ that does this successful \op{CAS} during $R_2$.
$R_3$ must have read $\var{v}.\fld{blocks}$ after the successful \op{CAS} during $R_1$.
The claim follows from \Cref{successfulRefreshGC} applied to $R_3$.
\end{proof}

\Cref{lem::appendExactlyOnce} can then be proved in the same way as in Section \ref{sec::propagating}.

\subsubsection{GC Keeps Needed Blocks}

The correctness of \op{GetEnqueue} and \op{IndexDequeue}, which are very similar to the 
original implementation, are dependent only on the fact that GC does not 
discard blocks needed by those routines.
The following results are used to show this.

We say a block is \emph{finished} if 
\begin{itemize}
\item
it has been propagated to the root,
\item
the value of each enqueue contained in the block has either been returned by a dequeue
or written in the \fld{response} field of a dequeue, and
\item
each dequeue contained in the block has terminated or some process has written to the \fld{response} field in the leaf block that represents it.
\end{itemize}
Intuitively, once a block is finished, operations no longer need  the block
to compute responses to operations.
The following is an immediate consequence of the definition of finished and what it means for an operation to be contained in a block.

\begin{observation}\label{subblocksFinished}
A block is finished if and only all of its subblocks are finished.
\end{observation}

\here{Technically, following proof might have to be wrapped up in a big induction with the argument
that outputs are linearized correctly, to argue that the matchup between enqueues and dequeues
is the same as it would be in the linearization ordering.  (Check whether this is really needed.)
Or alternatively, we could add an assumption that if every operation agrees with the linearization
so far, then the invariant is preserved by each step.
Then, when proving the linearization is correct, everything will work out.
But just leave it as is for the submission, since this is a subtle point}

\begin{invariant}\label{lem::finished}
If the minimum \fld{index} of any block in \var{v}'s RBT is $b_{min}$, then
each block with $\fld{index}$ at most $b_{min}$ that was ever added to \var{v} is finished.
\end{invariant}
\begin{proof}
The invariant is true initially, since the minimum \fld{index} block in \var{v}'s RBT is the empty block, which is (vacuously) finished.

We show that every step preserves the invariant.  
We need only consider a step $st$ that modifies a node's RBT.
The minimum \fld{index} of \var{v}'s RBT can only change when \var{v}'s RBT changes, either at line 
\ref{appendLeafGC} (if \var{v} is a leaf) or at line \ref{casGC} (if \var{v} is an internal node).
In either case, the step $st$ changes $\var{v}.\fld{blocks}$ 
from $T$ to $T'$, where $T'$ is obtained by  a call $A$ to \op{AddBlock}(\var{v}, $T$, $B$).
(In the case of a leaf \var{v}, this is true because only the process that owns the leaf ever
writes to $\var{v}.\fld{blocks}$.)
If $A$ does not do GC (lines \ref{GCstart}--\ref{GCend}), 
then $T'$ is obtained by adding a new block to $T$, so by \Cref{RBTupdates} the minimum \fld{index}
is unchanged and the invariant is trivially preserved.
So consider the case where $A$ performs GC.
We must show that any block of \var{v} whose index is less than or equal to the
minimum \fld{index} in $T'$ is finished when $T'$ is installed in $\var{v}.\fld{blocks}$.
Since $T'$ is obtained by discarding all blocks with \fld{index} at most $s$ (and adding
a block with a larger index), it suffices to show that all blocks that were ever added
to \var{v}'s RBT with \fld{index} at most $s$ are finished.

We must examine how $A$'s call to the recursive algorithm \op{SplitBlock} computes the value of $s$.

\begin{claim}\label{splitfinished}
If one of the recursive calls to $\op{SplitBlock}(x)$ within $A$'s call to \op{SplitBlock},
then that block and all earlier blocks in $x$ are finished when $st$ occurs.
\end{claim}
\begin{proof}[Proof of Claim]
We prove this claim  by induction on the depth of $x$.

For the base case, suppose $x$ is the root.
\op{SplitBlock} finds the maximum value $m$ in \var{last}, which is the index of some block
that contains an operation that is either a null dequeue or an enqueue whose value is the response for a dequeue that has been propagated to the root (since these are the only ways that an entry of \var{last} can be set to $m$).
By the FIFO property of queues, the values enqueued by enqueues in $\var{root}.\fld{blocks}[1..m-1]$
are all dequeued by operations that have already been installed in $\var{root}.\fld{blocks}$ before
the end of the \op{SplitBlock}.
Between the termination of $A$'s call to \op{SplitBlock} at line \ref{callSplitBlock}
and the \op{CAS} step $st$ after $A$ terminates,  $A$ helps all pending dequeues at line \ref{invokeHelp}.
Thus, after this helping (and before step $st$),
$\var{root}.\fld{blocks}[1..m-1]$ are all finished blocks, so the claim is true.

For the induction step, we assume the claim holds for $x$'s parent, and prove it for $x$.
We consider two cases.

If \op{SplitBlock}($x$) returns the minimum block of $x$'s RBT at line \ref{substitute},
then the claim follows from the assumption that \Cref{lem::finished} holds at all times before $st$.

Otherwise, \op{SplitBlock}($x$) returns the block $B$ at line \ref{substitute}.
By the induction hypothesis, the block $B_p$ computed at line \ref{SBrecurse}
(and all earlier blocks of $x$) are finished when $st$ occurs.
By \Cref{subblocksFinished}, the block $B$ in $x$ indexed by $B_p.\eleft$ or $B_p.\eright$
is also finished when $st$ occurs.  This completes the proof of \Cref{splitfinished}.
\renewcommand{\qedsymbol}{$\diamondsuit$}
\end{proof}

If the \op{Split} at line \ref{splitGC} of $A$ modifies the RBT, then it discards
all the blocks older than the one returned by \op{SplitBlock} at line \ref{callSplitBlock}.
By \Cref{splitfinished}, the minimum block in the new tree will satisfy the invariant when $st$
installs the new tree in $\var{v}.\fld{blocks}$.
\end{proof}

We remark that \Cref{lem::finished} guarantees that GC keeps one block that is finished and discards all blocks with smaller indices.
This is because the first unfinished block may still be traversed by an operation in the future,
and when examining that block the operation may need information from the preceding block.  For example,
when \op{FindResponse} is called on a block with index $b$, line \ref{FRnumenqGC} looks up the block with
index $b-1$.

\begin{lemma}\label{blocknotfound}
If a \op{Dequeue} operation fails to find a block in an RBT,
then it has been propagated to the root and its result has been written in its \fld{response} field.
If an \op{Enqueue} operation fails to find a block in an RBT, it has been propagated to the root.
\end{lemma}
\begin{proof}
Any block that GC removes from an RBT is finished, by \Cref{lem::finished}, so if an 
\op{Enqueue} or \op{Dequeue} fails to find a block while it is propagating itself up to
the root (for example, during the \op{CreateBlock} routine), then a block containing
the operation itself has been removed from a RBT, so the operation has propagated to the root,
by the definition of finished.
Moreover, by \Cref{lem::finished}, if the operation is a dequeue, then its result is in its \fld{response} field.

After propagation to the root, a \op{Dequeue} must access blocks that contain the dequeue,
as well as the enqueue whose value it will return (if it is not a null dequeue).
If any of those blocks have been removed, it follows from \Cref{lem::finished} that the
\op{Dequeue}'s result is written in its \fld{response} field.
\end{proof}

By \Cref{blocknotfound}, an operation that fails to find a block in an RBT can terminate.
If it is a \op{Dequeue}, it can return the result written in its \fld{response} field.
Since no other process updates the RBT in a process's leaf, the block containing the \fld{response}
will be the last block in the leaf's RBT, and the last block of an RBT is never removed by GC.
Thus, the \fld{response} field will still be there when a \op{Dequeue} needs it.

\subsubsection{Linearizability}

The correctness of the \op{IndexDequeue} and \op{GetEnqueue}  operations
can be proved in the same way as in Section \ref{sec::tracingCorrect},
since they are largely unchanged (except for the simplification that \op{IndexDequeue}
can simply search for a block's superblock instead of using the block's \fld{super} 
field to calculate the superblock's position).
They will give the correct response, provided none of the blocks they need to access have been 
removed by GC.  But as we have seen above, if that happens, the \op{Enqueue} or \op{Dequeue} can simply terminate.

Similarly, the results of Section \ref{sec::linearizability} 
can be reproved in exactly the same way as for the original algorithm to establish
that the space-bounded algorithm is linearizable.

\subsection{Analysis}
\label{sec::GCanalysis}

\here{need to prove a lemma analogous to \Cref{blockSize}}

We first bound the size of RBTs.  Let $q_{max}$ be the maximum size of the queue at any time during the sequential execution given by the linearization $L$.
Recall that GC is done on a node every $G$ times its RBT is updated, and we chose $G$ to be $p^2\ceil{\log p}$.
Part of the proof of the following lemma is similar to the proof of \Cref{dSearchTime}.

\begin{lemma}\label{boundAfterGC}
If the maximum \fld{index} in a node's RBT is a multiple of $G$, then
it contains at most $2q_{max}+4p+1$ blocks.
\end{lemma}
\begin{proof}
Consider the invocation $A$ of \op{AddBlock} that updates a node \var{v}'s RBT with the insertion of a block whose 
\fld{index} is a multiple of $G$.
Then, $A$ performs a GC phase.
Let $C$ be the configuration before $A$ invokes \op{SplitBlock} on line \ref{callSplitBlock}.
That call to \op{SplitBlock} recurses up to the root, where it computes $m$ by reading the \var{last}
array.
Let $L_1$ be the prefix of the linearization $L$ corresponding to blocks $1..m$ of the root.
Let $L_2$ be the next segment of the linearization corresponding to blocks $m+1..\ell$ of the root, where
$\ell$ is the last block added to the root's RBT before $C$.

We first bound the number of operations in $L_2$.

The number of enqueues in $L_2$ whose values are still in the queue at the end of $L_2$ is at most $q_{max}$.
Any enqueue in $L_2$ whose value is not still in the queue at the end of $L_2$
must still be in progress at~$C$; otherwise the process that dequeued it would have
set its \var{last} entry to a value greater than the index of the block of the root that contains the enqueue prior to $C$,
contradicting the definition of $m$.
So, there are at most $p$ enqueues in $L_2$ whose values are still in the queue at the end of $L_2$.
Thus, there are at most $q_{max}+p$ enqueues in $L_2$.

If a dequeue in $L_2$ returns a non-null value in the sequential execution $L$,
then the value it returns was either in the queue at the end of $L_1$ or it was enqueued during $L_2$.
Thus, there are at most $q_{max}+(q_{max}+p)$ dequeues in $L_2$ that return non-null values.
Any dequeue in $L_2$ that returns a null value in the sequential execution $L$ must still be in progress
at $C$; otherwise the process that performed the dequeue would have set its \var{last} entry
to a value greater than the index of the block of the root that contains the dequeue prior to $C$, 
contradicting the definition of $m$.
So, there are at most $p$ null dequeues in $L_2$.
Thus, there are at most $2q_{max} + 2p$ dequeues in $L_2$.

$A$'s call to \op{SplitBlock} determines the index $s$ used to split \var{v}'s RBT by following 
\eleft\ and \eright\ pointers from the root down to \var{v}.
So, the block returned is a subblock of block $m$ of the root, unless
at some point along the path of subblocks the subblock has already been removed by a split,
in which case \op{SplitBlock} returns a subblock of a block $m'>m$ in the root.

Next, we bound the number of operations in \var{v}'s blocks that are retained when $A$  sets
$T'\assign\op{Split}(T,s)$.
Since $T$ was read before $C$, any operation in $T$ is either in progress at $C$ or has been propagated to the root before $C$.\here{cite a lemma about propagation before operation terminates}
Thus, there are at most $p$ operations in $T$ that do not appear in $L_1\cdot L_2$.
All the rest of the operations in blocks of $T'$ have been propagated to blocks $m..\ell$ of the root.
There are at most $p$ operations in block $m$ of the root 
and we showed above that there are at most $2q_{max} + 2p$ in blocks 
$m+1..\ell$ of the root.
Thus, there are at most $2q_{max}+4p$ operations in blocks of $T'$.
Since each block is non-empty by \Cref{blockNotEmpty}, $T'$ contains at most $2q_{max}+4p$ blocks, and one more 
block is inserted before $A$ sets \var{v}'s \fld{blocks} to the resulting RBT.
\end{proof}

\begin{corollary}\label{RBTbound}
At all times, the size of a node's RBT is $O(q_{max}+p+G)$. 
\end{corollary}
\begin{proof}
Each update to a node's RBT adds at most one block to it, increasing its maximum \fld{index} by 1.
Thus, there are at most $G$ updates since the last time its maximum \fld{index} was a multiple of $G$.
The claim follows from \Cref{boundAfterGC}.
\end{proof}

The following theorem bounds the space that is reachable (and therefore cannot be freed by the environment's garbage collector) at any time.
\here{For later version:  There may be room to improve this bound by a more careful analysis since the bounds so far
didn't use the fact that operations only appear along one path in the tree; in other words
doing a bound on total size of all RBTs at one level of the ordering tree at once, rather than
on each node separately, might reduce it to qlogp + poly(p).  Might also be possible to be
more precise about how it depends on "current" size of queue rather than max size.}

\begin{mytheorem}\label{spaceBound}
The queue data structure uses a maximum of $O(pq_{max}+p^3\log p)$ words of memory at any time.
\end{mytheorem}
\begin{proof}
There are $2p-1$ nodes in the ordering tree.  Aside from the RBT, each node uses $O(1)$ memory words.
Each process may hold pointers to $O(1)$ RBTs that are no longer current in local variables.
So the space bound follows from \Cref{RBTbound} and the fact that $G$ is chosen
to be $p^2\ceil{\log p}$.
\end{proof}

\here{It may be possible, via a more careful analysis, to reduce the bound on space usage to $O(\log^2 p(q_{max}+p^3)$.}

Although individual operations may now become more expensive because they have to perform GC,
we show that operations still have polylogarithmic amortized step complexity.

\begin{mytheorem}
The amortized step complexity of each operation is $O(\log p \log(p+q_{max}))$.
\end{mytheorem}
\begin{proof}
\here{this proof may need some modification since code was changed to change how split point is computed}
It follows from \Cref{RBTbound} and our choice of $G=p^2\ceil{\log p}$
that all the RBT routines we use to perform \op{Split}, \op{Insert} and searches for
blocks with a particular index or for a \fld{sum\sub{enq}} value (in line \ref{FRsearchGC} or \ref{getChildGC}) can  be done in $O(\log(p+q_{max}))$ steps.

First, we bound the number of steps taken \emph{excluding} the GC phase in line \ref{GCstart}--\ref{GCend}.
An \op{Enqueue} or null \op{Dequeue} does $O(1)$ RBT operations and other work at each level of the tree during \op{Propagate},
for a total of $O(\log p \log(p+q_{max}))$ steps.
A non-null \op{Dequeue} must also search for a block in the root at line \ref{FRsearchGC}
and call \op{GetEnqueue}.  At each level of the tree, \op{GetEnqueue} does $O(1)$ RBT operations (including a search at line \ref{getChildGC}) and $O(1)$ other steps.
Thus, a \op{Dequeue} also takes $O(\log p \log(p+q_{max}))$ steps.

Now we consider the additional steps a process takes while doing GC in line \ref{GCstart}--\ref{GCend}
and show that the amortized number of GC steps each operation performs is also $O(\log p\log(p+q_{max}))$.
If a process does GC in a call to \op{AddBlock}($\var{v},T,B$) where $B$ has \fld{index} $r\cdot G$ for some integer $r$, we call this the process's \emph{$r$th helping phase on \var{v}}.

We argue that each process $P$ can do an $r$th helping phase on \var{v} at most once.
Consider $P$'s first call $A$ to \op{AddBlock} that does an $r$th helping phase on \var{v}.
Let $\var{v},T,B$ be the arguments of $A$.
Any call to \op{AddBlock} on internal node \var{v} is from line \ref{refreshABGC} of \op{Refresh}, so $B.\fld{index}$
is the maximum \fld{index} in $T$ plus 1.
The \op{Refresh} that called $A$ performed a \op{CAS} at line \ref{casGC}.  Either the \op{CAS} succeeds
or it fails because 
some other \op{CAS} changes \var{v}.\fld{blocks} from $T$ to another tree.
Either way, by \Cref{RBTupdates}, \var{v}'s RBT's maximum \fld{index} will be at least $r\cdot G$ at all times
after this CAS.
So if a subsequent \op{Refresh} by process $P$ 
ever calls \op{AddBlock} on \var{v} again, the block it passes as the third argument
will have $\fld{index}>r\cdot G$, so $P$ will not perform an $r$th helping phase on \var{v} again.

Each helping phase takes $O(p)$ steps to read the \var{last} array, and
$O(p \log p \log(p+q_{max}))$ steps in \op{Help},
and $O(\log(p+q_{max}))$ steps to split and insert a new node into a RBT.
Thus, for each integer $r$ and each node \var{v}, a total of $O(p^2\log p\log(p+q_{max}))$ steps
are performed by all processes during their $r$th helping phase on \var{v}.
We can amortize these steps over the operations that appear in 
$\var{v}.\fld{blocks}[(r-1)G+1..rG]$.
By \Cref{blockNotEmpty}, there are at least $G$ such operations, 
so each operation's amortized number of steps for GC at each node along the path from its leaf to the root
is $O(p^2\log p\log(p+q_{max})/G)=O(\log(p+q_{max}))$.
Hence each operation's amortized number of GC steps is $O(\log p\log(p+q_{max}))$.
\end{proof}

The implementation remains wait-free:  the depth of recursion in each routine is  bounded
by the height of the tree and the only loop is the counted loop in the \op{Help} routine.
Since each operation still  does only two \op{CAS} instructions at each level of the tree (at line \ref{casGC}), the following proposition still holds for the space-bounded version of the queue.

\begin{customprop}{\bf{\ref{casbound}}$'$}
Each operation performs $O(\log p)$ \op{CAS} instructions in the worst case.
\end{customprop}

\end{document}